\documentclass[reqno,onecolumn,letterpaper,11pt]{amsart}

\usepackage{comment,graphicx,amsmath,amsthm, subfigure,url}
\usepackage{enumitem}
\usepackage{todonotes}
\usepackage{amssymb}
\usepackage{xspace}

\newcommand{\II}{\mathcal{I}}
\newcommand{\Real}{\mathbb{R}}
\newcommand{\ZZ}{\mathbb{Z}}
\newcommand{\RR}{\mathcal{R}}

\theoremstyle{plain}
\newtheorem{Theo}{Theorem}
\newtheorem{Claim}{Claim}
\theoremstyle{definition}
\newtheorem{Def}{Definition}
\newtheorem{Prop}{Property}

\newtheorem{Coro}{Corollary}

\newtheorem{Lem}{Lemma}
\theoremstyle{remark}

\DeclareMathOperator{\wmis}{wmis}
\DeclareMathOperator{\mis}{mis}
\DeclareMathOperator{\mhs}{mhs}
\DeclareMathOperator{\LP}{LP}

\newcommand{\OPT}{\ensuremath{\mathrm{OPT}}\xspace}
\newcommand{\MHS}{\ensuremath{\mathrm{MHS}}\xspace}
\newcommand{\MIS}{\ensuremath{\mathrm{MIS}}\xspace}
\newcommand{\WMIS}{\ensuremath{\mathrm{WMIS}}\xspace}
\newcommand{\dg}{\delta_{\text{GAP}}}
\providecommand{\pb}[1]{{\sc #1} problem}

\title[Rectangles Intersecting a Diagonal Line]{Independent and Hitting Sets of Rectangles Intersecting a Diagonal Line:\\ Algorithms and Complexity.}
\author[J.R. Correa]{Jos\'e R.~Correa$^1$}
\author[L. Feuilloley]{Laurent Feuilloley$^2$}
\author[P. P\'erez-Lantero]{Pablo P\'erez-Lantero$^3$}
\author[J.A. Soto]{Jos\'e A.~Soto$^4$}
\thanks{An extended 
abstract of a preliminary version of this work appears in the proceedings of the 11th Latin American Theoretical Informatics Symposium.}
\dedicatory{
$^1$Department of Industrial Engineering, Universidad de Chile.\\ {\tt{correa@uchile.cl}}
\\[1ex]
$^2$Department of Computer Science, ENS Cachan.\\ {\tt laurent.feuilloley@ens-cachan.fr}\\[1ex]
$^3$Escuela de Ingenier\'ia Civil en Inform\'atica, Universidad de Valpara\'{i}so. {\tt pablo.perez@uv.cl}\\[1ex]
$^4$Department of Mathematical Engineering and CMM, Universidad de Chile.\\ {\tt jsoto@dim.uchile.cl}}
\begin{document}
\maketitle
\thispagestyle{empty}

\begin{abstract}Finding a maximum independent set (MIS) of a given family of axis-parallel rectangles is a
basic problem in computational geometry and combinatorics. This problem has attracted significant attention since the sixties, when Wegner conjectured that the corresponding duality gap, i.e., the maximum possible ratio between the maximum independent set and the minimum hitting set (MHS), is bounded by a universal constant. An interesting special case, that may prove useful to tackling the general problem, is the diagonal-intersecting case, in which the given family of rectangles is intersected by a diagonal. Indeed, Chepoi and Felsner recently gave a factor 6 approximation algorithm for MHS in this setting, and showed that the duality gap is between 3/2 and 6. In this paper we improve upon these results. First we show that MIS in diagonal-intersecting families is NP-complete, providing one smallest subclass for which MIS is provably hard. Then, we derive an $O(n^2)$-time algorithm for the maximum weight independent set when, in addition the rectangles intersect below the diagonal. This improves and extends a classic result of Lubiw, and amounts to obtain a 2-approximation algorithm for the maximum weight independent set of rectangles intersecting a diagonal. Finally, we prove that for diagonal-intersecting families the duality gap is between 2 and 4. The upper bound, which implies an approximation algorithm of the same factor, follows from a simple combinatorial argument, while the lower bound represents the best known lower bound on the duality gap, even in the general case.
\end{abstract}

\newpage

\section{Introduction}
Given a family of axis-parallel rectangles, two natural objects of
study are the maximum number of rectangles that do not overlap and the
minimum set of points stabbing every rectangle. These problems are
known as maximum independent set \MIS and minimum hitting set \MHS respectively,
and in the associated intersection graph they correspond to the
maximum independent set and the minimum clique covering. We study these problems for restricted classes of rectangles, and
focus on designing algorithms and on evaluating the {\em duality gap}, $\dg$,
i.e., the maximum ratio between these quantities. This term arises as \MHS is the integral version of the dual of
the natural linear programming relaxation of~\MIS.

From a computational complexity viewpoint, \MIS and \MHS of rectangles are strongly NP-hard~\cite{Fowler81,Imai83}, so attention has been put into approximation algorithms and polynomial time algorithms for special classes. The current best known approximation factor for \MIS  are $O(\log\log n)$ \cite{Chalermsook09}, and $O(\log n / \log\log n)$ for weighted \MIS (\WMIS)~\cite{Chan09}. Very recently, Adamaszek and Wiese~\cite{Adamaszek13} designed a pseudo-polynomial time algorithm finding a $(1+\varepsilon)$-approximate solution for \WMIS, but it is unknown whether there exist polynomial time constant factor approximation algorithms. A similar situation occurs for \MHS: the current best approximation factor is
$O(\log\log n)$ \cite{Aronov10}, while in general, the existence of a constant factor approximation is open. Polynomial time algorithms for these
problems have been obtained for special classes. When all rectangles are intervals, the underlying intersection graph is an interval graph and even
linear time algorithms, assuming the input is sorted, are known for \MIS, \MHS and \WMIS~\cite{HsiaoT92}. Moving beyond interval graphs, Lubiw~\cite{Lubiw91} devised a
cubic-time algorithm for computing a maximum weight independent family of point-intervals, which can be seen as families of rectangles having their
upper-right
corner along the same diagonal.  More recently, Soto and Telha~\cite{SotoTelha11} considered the case where the upper-right and lower-left corners of
all rectangles are two prescribed point sets of total size $m$. They designed an algorithm that computes both \MIS and \MHS in the time required to
do $m$ by $m$ matrix multiplication, and showed that \WMIS is NP-hard on this class. Finally, there are also known PTAS for special cases, including
the results of
Chan~\cite{Chan09} for squares, and Mustafa and Ray~\cite{Mustafa10} for unit height rectangles.

It is straightforward to observe that given a family of rectangles the size of a maximum independent set is at most that of a minimum
hitting set. In particular, for interval graphs this inequality is actually an equality,
and this still holds in the case studied by Soto and Telha~\cite{SotoTelha11}, so that the duality gap is 1 for these classes. A natural question to
ask
is whether the duality gap for general families of rectangles is bounded.  Indeed, already in the sixties Wegner~\cite{Wegner65} conjectured that the
duality gap for arbitrary
rectangles families equals~2, whereas Gy\'arf\'as and Lehel~\cite{GyarfasL85} proposed the weaker conjecture that this gap is bounded by a universal
constant. Although
these conjectures are still open, K\'arolyi and Tardos~\cite{Karolyi96} proved that the gap is within $O(\log( \mis))$, where $\mis$ is the size of a
maximum independent set. For some special classes, the duality gap is indeed a constant. In particular, when all rectangles intersect a given
diagonal
line, Chepoi and Felsner \cite{ChepoyF13} prove that the gap is between $3/2$ and $6$, and the upper bound has been further improved for more
restricted classes~\cite{ChepoyF13,Hixon13}.

\subsection{Notation and classes of rectangle families}
Throughout this paper, $\RR$ denotes a family of $n$ closed, axis-parallel rectangles in~$\Real^2$. A rectangle $r\in \RR$ is defined by its lower-left
corner $\ell^r$ and its upper-right corner $u^r$. For a point $v\in \Real^2$ we let $v_x$ and $v_y$ be its $x$-coordinate and $y$-coordinate,
respectively. Also, each rectangle $r\in \RR$ is associated with a nonnegative weight $w_r$. We also consider a monotone curve, given by a decreasing
bijective real function, so that the boundary of each $r \in \RR$ intersects the curve in at most 2 points. We use $a^r$ and $b^r$ to denote the
higher
and lower of these points respectively (which may coincide). We identify the
rectangles in $\RR$ with the set
$[n]=\{1,\ldots,n\}$ so that  $a^1_x< a^2_x < \dots < a^n_x$. For any rectangle~$i$, we define $f(i)$ as the rectangle $j$ (if it
exists) following $i$  in the
order of the $b$-points, that is, $b^i_x < b^j_x$ and no rectangle $k$ is such that $b^i_x < b^k_x< b^j_x$. For reference, see
Figure~\ref{fig:harpoons}.

A set of rectangles $\mathcal{Q}\subseteq \RR$ is called independent if and only if no two rectangles in
$\mathcal{Q}$ intersect. On the other hand, a set $H \subseteq \Real^2$  of points is a hitting set of~$\RR$ if every rectangle $r\in \RR$ contains at
least one point in~$H$.
In this paper we consider the problem of finding an independent set of rectangles in $\RR$ of maximum cardinality (\MIS), and its weighted version
(\WMIS). We also consider the problem of finding a hitting set of $\RR$ of minimum size (\MHS). Let us denote by $\mis(\RR)$, $\wmis(\RR)$,
$\mhs(\RR)$ the solutions to the above problems, respectively.

Since the solutions of the previous problems depend on properties of the intersection graph $\II(\RR)=(\RR,\{rr'\colon r\cap r'\neq \emptyset\})$
 of the family $\RR$, we will assume that no two defining corners in $\{\ell^1,\ell^2,\ldots,\ell^n, u^1, u^2,\ldots, u^n\}$ have the same
$x$-coordinates or $y$-coordinates (this is done without loss of generality by individually perturbing each rectangle). We will also
assume that the curve mentioned in the first paragraph is the diagonal line $D$ given by the equation $y=-x$. This is assumed without loss of
generality: by
applying suitable piecewise linear transformations on both coordinates we can transform the rectangle family into one with the same intersection graph
such that every rectangle intersects $D$. In what follows, call the closed halfplanes given by $y\geq -x$ and $y \leq -x$, the \emph{halfplanes} of $D$. Note that both halfplanes intersect in
$D$. The points in the bottom (resp.~top) halfplane are said to be below (resp.~above) the diagonal.

We study four special classes of rectangle families intersecting $D$.
\begin{Def}[Classes of rectangle families]\mbox{}\label{def-classes}
\begin{enumerate}[leftmargin=17pt]
\item $\RR$ is \emph{diagonal-intersecting} if for all $r\in \RR$, $r\cap D\neq \emptyset$.
\item $\RR$ is \emph{diagonal-splitting} if there is a side (upper, lower, left, right) such that $D$ intersects all $r \in R$ on that particular
side.
\item $\RR$ is \emph{diagonal-corner-separated} if there is a halfplane of $D$ containing the same three corners of all $r\in \RR$.
\item $\RR$ is \emph{diagonal-touching} if there is a corner (upper-right or lower-left) such that $D$ intersects all $r \in R$ exactly on that
corner (in particular, either all the upper-right corners, or all the lower-left corners are in $D$.)
\end{enumerate}
\end{Def}
By rotating the plane, we can make the following assumptions: In the second class, we assume that the common side of intersection is the upper one; in
the third class, that the upper-right corner is on the top halfplane of $D$ and the other three are in the bottom one; and in the last
class, that the corner contained in $D$ is the upper-right one. Under these assumption, each type of rectangle family is more
general than the next one. It is worth noting that in terms of their associated intersection graphs, the second and third classes coincide. Indeed, two rectangles of a diagonal-splitting rectangle family $\RR$ intersect if and only if they have a point in common in the bottom halfplane of $D$. Therefore, we can replace each rectangle $r$ with the minimal possible one containing the region of~$r$ that is below the diagonal, obtaining a diagonal-corner-separated family with the same intersection graph. See Figure~\ref{fig:cases_figures} for some examples of rectangle families.

\begin{figure}
\centering
\begin{tabular}{ccc}
\subfigure{\includegraphics[scale=0.9,page=1]{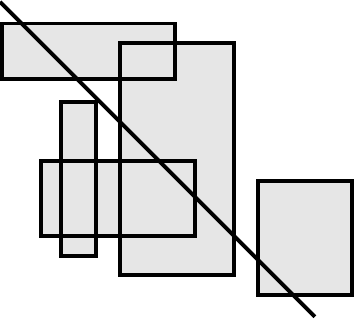}}&\qquad
\subfigure{\includegraphics[scale=0.9,page=2]{five_cases}}&\qquad
\subfigure{\includegraphics[scale=0.9,page=3]{five_cases}}\\
\small diagonal-intersecting & \small diagonal-lower-intersecting &\small diagonal-splitting\\
\end{tabular}
\begin{tabular}{cc}
\subfigure{\includegraphics[scale=0.9,page=4]{five_cases}}&\qquad
\subfigure{\includegraphics[scale=0.9,page=5]{five_cases}}\\
\small diagonal-corner-separated & \small diagonal-touching
\end{tabular}
\caption{Examples of rectangle families.}\label{fig:cases_figures}
\end{figure}

\begin{Def}[diagonal-lower-intersecting]
A diagonal-intersecting family $\RR$ is \emph{diagonal-lower-intersecting} if whenever two rectangles in $\RR$ intersect, they have a common point in
the bottom halfplane of $D$.
\end{Def}

As we will see later, the graph classes associated to these families satisfy the following inclusions: $\mathcal{G}_{\text{touch}} \subsetneq \mathcal{G}_{\text{low-int}} = \mathcal{G}_{\text{split}} = \mathcal{G}_{\text{c-sep}} \subsetneq \mathcal{G}_{\text{int}}$.
Here $\mathcal{G}_{\text{int}}=\{\II(\RR)\colon \RR \text{ is diagonal-intersecting}\}$ is the class of intersection graphs arising from diagonal-intersecting families of rectangles, and $\mathcal{G}_{\text{low-int}}$, $\mathcal{G}_{\text{split}}$, $\mathcal{G}_{\text{c-sep}}$ and $\mathcal{G}_{\text{touch}}$ are the classes arising from diagonal-lower-intersecting,
diagonal-splitting, diagonal-corner-separated, and diagonal-touching families of rectangles, respectively.
We observe that these classes have appeared in the literature under different names. Hixon~\cite{Hixon13} call the graphs
in $\mathcal{G}_{\text{touch}}$ \emph{hook graphs}, Soto and Thraves~\cite{MSoto13} call them \emph{\textsc{And(1)} graphs}, while those in $\mathcal{G}_{\text{int}}$ are called \emph{separable rectangle graphs} by Chepoi and Felsner~\cite{ChepoyF13}.

\subsection{Our results}
In $\S 2$ we give a quadratic-time algorithm to compute a $\wmis(\RR)$ when $\RR$ is diagonal-lower-intersecting and a
2-approximation for the same
problem when $\RR$ is diagonal-intersecting. The former is the first polynomial time algorithm for \WMIS on a natural class
containing diagonal-touching rectangle families. Our algorithm improves upon previous work in the area. Specifically, for diagonal-touching rectangle
families, the best known algorithm to solve \WMIS is due to Lubiw~\cite{Lubiw91}, who designed a cubic-time algorithm for the problem in the context
of
\emph{interval systems}. More precisely, a collection of \emph{point-intervals} $Q=\{(p_i,I_i)\}_{i=1}^n$ is a family such that for all $i$, $p_i \in
I_i$ and $I_i=[\text{left}(I_i), \text{right}(I_i)] \subseteq \Real$ are a point
 and an interval, respectively. $Q$ is called \emph{independent} if for $k\neq j$, $p_k \notin I_j$ or $p_j \notin I_k$. Given a finite collection
$Q$ of weighted point-intervals, Lubiw designed a dynamic programming based algorithm to find a maximum weighted independent subfamily of $Q$. It is
easy to see\footnote{This equivalence has been noticed before~\cite{SotoTelha11}.} that this problem is equivalent to that of
finding $\wmis(\RR)$ for the diagonal-touching family $\RR=\{r_i\}_{i=1}^n$ where $r_i$ is the rectangle with upper right corner $(p_i,-p_i)$ and
lower left corner $(\text{left}(I_i),-\text{right}(I_i))$ and having the same weight as that of $(p_i,I_i)$. Lubiw's algorithm was recently
rediscovered by Hixon \cite{Hixon13}.

As in Lubiw's, our algorithm is based on dynamic programming. However, rather than decomposing the instance into small triangles and
computing the optimal solution for every possible triangle, our approach involves computing the optimal solutions for what we call a \emph{harpoon},
which is defined for every pair of rectangles. We show that the amortized cost of computing the optimal solution for all harpoons is constant, leading
to an overall quadratic time. Interestingly, it is possible to show that our algorithm is an extension of the linear-time algorithm for
maximum weighted independent set of intervals~\cite{HsiaoT92}.

 In \S\ref{duality} we give a short proof that the duality gap $\dg$, i.e., the maximum ratio $\mhs/\mis$, is always at most 2 for
diagonal-touching families; we also show that $\dg \leq 3$ for diagonal-lower-intersecting families, and $\dg \leq 4$ for diagonal-intersecting
families. These bounds yields simple 2, 3, and 4-approximation polynomial time algorithms for \MHS on each class (they can also be used as
approximation algorithms for \MIS with the same guarantee, however, as discussed in the previous paragraph, we have an exact algorithm for \WMIS on the
two first classes, and a 2-approximation for the last one). The 4-approximation for \MHS in diagonal-intersecting families is the best
approximation known and improves upon the bound of 6 of Chepoi and Felsner~\cite{ChepoyF13}, who also give a bound of 3 for diagonal-splitting
families based on a different method. For diagonal-touching families, Hixon~\cite{Hixon13} independently showed that
$\dg \leq 2$. To complement the previous results, we show that the duality gap for diagonal-lower-intersecting families is at least 2. We do this
by exhibiting an infinite family of instances whose gap is arbitrarily close to 2. Similar instances were obtained, and communicated to us, by Cibulka
et al.~\cite{Cibulka06}. Note that this lower bound of 2 improves upon the 5/3 by Fon-Der-Flaass and
Kostochka~\cite{FonDerFlaassK93} which was the best known lower bound for the duality gap of general rectangle families.

In \S\ref{sec:hardness}, we prove that computing a $\MIS$ on a diagonal-intersecting family is NP-complete.
In light of our polynomial-time algorithm for diagonal-lower-intersecting families, the latter
hardness result exhibits what is, in a way, a class at the boundary between polynomial-time solvability and NP-completeness. Three decades ago Fowler 
et al.~\cite{Fowler81} (see also Asano~\cite{Asano91}) established that computing an MIS of axis-parallel rectangles squares is NP-hard, by actually 
showing that this is the case even for squares. It is worth mentioning that diagonal-intersecting families constitute the first natural subclass of 
for which NP-hardness of MIS has been shown since then. Our proof 
actually only uses rectangles that touch the diagonal line, but that may intersect above or below it, and uses a reduction from  {\sc Planar 3-sat}.

Combining the results of  Chalermsook and Chuzhoy~\cite{Chalermsook09} and Aronov et al.~\cite{Aronov10}, we show in \S\ref{general-duality} that the duality gap is $O((\log\log \mis(\RR))^2)$ for a general family $\RR$ of rectangles, improving on the logarithmic bound of K\'arolyi and Tardos~\cite{Karolyi96}. Finally, in \S\ref{graph-classes} we prove the claimed inclusions of the rectangle families studied in this paper, described in Definition \ref{def-classes}.

\section{Algorithms for \WMIS}

The idea behind Lubiw's algorithm \cite{Lubiw91} for \WMIS on diagonal-touching families is to compute the optimal independent set $\OPT_{ij}$
included in every possible triangle defined by the points $u^i$, $u^j$ (which are on $D$), and $(u^i_x,u^j_y)$ for two rectangles $i<j$. The principle
exploited is that in $\OPT_{ij}$ there exists one rectangle, say $i<k<j$, such that $\OPT_{ij}$ equals the union of $\OPT_{ik}$, the rectangle $k$,
and $\OPT_{kj}$. With this idea the overall complexity of the algorithm turns out to be cubic in $n$. We now present our algorithm, which works for
the more general diagonal-lower-intersecting families, and that is based in a more elaborate idea involving what we call \emph{harpoons}.

\subsection{Algorithm for diagonal-lower-intersecting families}

\label{algorithm_harpoons}
Let us first define some geometric objects that will be used in the algorithm. For any pair of rectangles $i<j$ we define $H_{i,j}$ and $H_{j,i}$, two
shapes that we call harpoons. See Figure~\ref{fig:harpoons}.
More precisely, the \emph{horizontal harpoon} $H_{i,j}$ consists of the points below the diagonal $D$ obtained by subtracting rectangle $i$ from the
closed box defined by the points $(\ell^i_x,a^i_y)$ and $a^j$.
Similarly, the \emph{vertical harpoon} $H_{j,i}$ are the points below $D$ obtained by subtracting $j$ from the box defined by the points
$(b^j_x,\ell^j_y)$ and $b^i$.
Also, for every rectangle $i$ with $i\geq 1$ (resp.~such that $f(i)$ exists) we define $B_{h}^i$ (resp.~$B_{v}^i$) as the open horizontal strip that goes through
$a_{i-1}$ and
$a_{i}$ (resp.~as the open vertical strip that goes through $b_{i}$ and $b_{f(i)}$).

\begin{figure}[ht]
\centering
\begin{tabular}{ccc}
\subfigure{\includegraphics[scale=0.8,page=1]{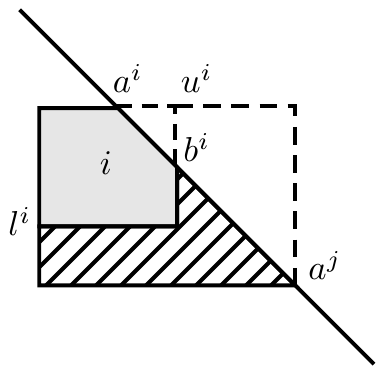}}&\qquad
\subfigure{\includegraphics[scale=0.8,page=2]{harpoons}}&\qquad
\subfigure{\includegraphics[scale=0.8,page=3]{harpoons}}\\
\subfigure{\includegraphics[scale=0.7,page=4]{harpoons}}&\qquad
\subfigure{\includegraphics[scale=0.8,page=5]{harpoons}}&\qquad
\subfigure{\includegraphics[scale=0.8,page=2]{harpoons}}
\end{tabular}
\caption{\small On the left, the construction of a harpoon and the construction of the strips. On the middle, the harpoons $H_{ij}$ and $H_{ji}$, with $i<j$.
On the right, other particular cases for the harpoon $H_{ij}$ with $i<j$ (the symmetric cases occur for $H_{ji}$).}\label{fig:harpoons}
\end{figure}

We say that a rectangle $r$ is contained in the set $H_{i,j}$ (and abusing notation, we write $r\in H_{i,j}$) if the region of $r$ below
the diagonal is contained in~$H_{i,j}$.

In our algorithm we will compute $S(i,j)$, the weight of the maximum independent set for the subset of rectangles contained in the harpoon
$H_{i,j}$. We define two dummy rectangles $0$ and $n+1$, at the two ends of the diagonal such that the harpoons defined by these rectangles contain
every other rectangle. As previously observed, two rectangles intersect in $\RR$ if and only if they intersect below the diagonal. Therefore,
$\mbox{wmis}(\mathcal{R})=S(0,n+1)$.

\paragraph{\textbf{Description of the algorithm:}}

\begin{itemize}[leftmargin=17pt]
\item[1.] \emph{Initialization.} In the execution of the algorithm we will need to know what rectangles have their lower-left corner in which strips.
To compute this we do a preprocessing step. Define $\hat B_{v}^i$ and $\hat B_{h}^i$ as initially empty. For each rectangle $r\in \RR$, check if
$\ell^r$ is in $B_{h}^i$. If so, we add $r$ to the set $\hat B_{h}^i$. Similarly, if $\ell^r$ is in $B_{v}^i$, we add $r$ to the
set $\hat B_{v}^i$.

\item[2.] \emph{Main loop.} We compute the values $S(i,j)$ corresponding to the maximum-weight independent set of rectangles in $\RR$ strictly
contained in $H_{i,j}$. We do this by dynamic programming starting with the values $S(i,i)=0$. Assume that we have computed all $S(i,j)$ for
all $i$, $j$ such that $|i-j| < \ell$. We now show how to compute these values when $|i-j|=\ell$.
\begin{itemize}
\item[2.1] Set  $S(i,j)=S(i,j-1)$ if $i<j$ and  $S(i,j)=S(i,f(j))$ if $i>j$.
\item[2.2] Define $\hat{B}_{i,j}$ as $\hat{B}_h^j$ if $i<j$, or $\hat{B}_v^j$ if $i>j$.
\item[2.3] For each rectangle $k \in \hat{B}_{i,j}$ and strictly contained in harpoon $H_{i,j}$ do:
\begin{itemize}
\item[2.3.1.] Compute $m = w_k + \max\{{S}(i,k),{S}(k,i)\} + S(k,j).$
\item[2.3.2.] If $m > S(i,j)$, then $S(i,j):= m$.
\end{itemize}
\end{itemize}
\item[3.] {\em Output.} ${S}(0,n+1)$.
\end{itemize}

It is trivial to modify the algorithm to return not only $\wmis(\RR)$ but also the independent set of rectangles attaining that weight.
We now establish the running time of our algorithm.

\begin{Theo}
The previous algorithm runs in $O(n^2)$.
\end{Theo}
\begin{proof}

The pre-processing stage needs linear time if the rectangles are already sorted, otherwise we require $O(n \log n)$ time. The time to compute
$S(i,j)$ is $O(1+|\hat{B}_{i,j}|)$ since checking if a rectangle is in a harpoon takes
constant time.
As the index of a rectangle is at most once in some $\hat{B}_h$ and at most once in some $\hat{B}_v$, the time to fill all the table $S(\cdot
,\cdot)$ is:
$$\sum_{(i,j)\in [n]^2}O(1+|\hat{B}_{i,j}|)= O(n^2).$$

The algorithm is then quadratic in the number of rectangles.
\end{proof}

In order to analyze the correctness of our algorithm we define a partial order over the rectangles in $\RR$.

\begin{Def}
The (strict) \emph{onion ordering} $\prec$ in $\RR$ is defined as
\begin{align*}
 i \prec j &\iff \text{rectangles $i$ and $j$ are disjoint}, \ell^i_x < \ell^j_x, \text{ and } \ell^i_y < \ell^j_y.
\end{align*}
\end{Def}
It is immediate to see that $\prec$ is a strict partial ordering in $\RR$. We say that $i$ is dominated by $j$ if $i\prec j$; in other words, $i$ is dominated by $j$ if $i$ and $j$ are disjoint and $\ell^i$ is dominated by $\ell^j$ under the standard dominance relation of $\mathbb{R}^2$.

For any rectangle $k$ in a harpoon $H_{i,j}$, let $S_k(i,j)$ be the value of the maximum-weight independent set containing  $k$ and rectangles in
$H_{i,j}$ which are not dominated by $k$ in the onion ordering, and $\mathcal{S}_k(i,j)$ be the corresponding set of rectangles.

\begin{Lem}
For any rectangle $k$ in $H_{i,j}$, the following relation holds:
\begin{align*}
S_{k}(i,j) = w_k + \max \left\{S(i,k),S(k,i)\right\} + S(k,j).
\end{align*}
\end{Lem}

\begin{proof}
Since $k \in H_{i,j}$, we have that $i$, $k$ and $j$ are mutually non-intersecting, and as indices, $\min(i,j)<k<\max(j,i)$.
Assume that the harpoon is horizontal, i.e., $i<j$ (the proof for $i>j$ is analogous). In particular, we know that $a^i, b^i, a^k, b^k,
a^j, b^j$
appear in that order on the diagonal.
There are three cases for the positioning of the two rectangles $i$ and $k$. See Figure~\ref{Three_cases}.

\begin{figure}[h]
\centering
\subfigure{\includegraphics[scale=0.7,page=1]{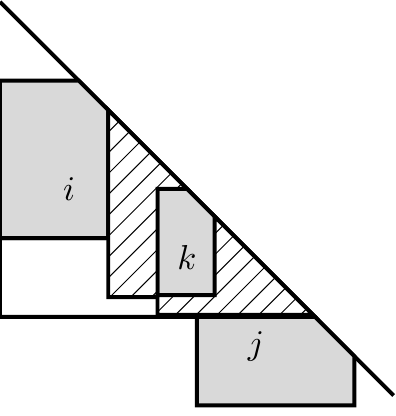}}
\subfigure{\includegraphics[scale=0.7,page=2]{fig2_ver1}}
\subfigure{\includegraphics[scale=0.7,page=3]{fig2_ver1}}
\caption{\small The three cases for a rectangle in a horizontal harpoon.}\label{Three_cases}
\end{figure}
\textit{First case}: $i$ and $k$ are separated by a vertical line, but not separated by a horizontal one. Noting that $H_{i,k} \subseteq H_{k,i}$,
we conclude that all the rectangles of
$\mathcal{S}_{k}(i,j)\setminus\{k\}$ are in $H_{k,i}$ or in $H_{k,j}$. Since $H_{k,i}$ and $H_{k,j}$ are disjoint, as shown on the first picture, we
conclude
the correctness of the formula.

\medskip

\textit{Second case}: $i$ and $k$ are separated by a horizontal line, but not by a vertical one. The proof follows almost exactly as in the first
case.

\medskip

\textit{Third case}: $i$ and $k$ are separated by both a horizontal line and a vertical line. By geometric and minimality arguments, all the
rectangles in $\mathcal{S}_k(i,j)\setminus\{k\}$ are in the union of the three harpoons $H_{i,k}$, $H_{k,i}$ and $H_{k,j}$
depicted. Finally, if there are two
rectangles in $H_{i,k} \cup H_{k,i}$ then they must be in the same harpoon, so the formula holds.
\end{proof}

\begin{Theo}
Our algorithm returns a maximum weight independent set of~$\RR$.
\end{Theo}

\begin{proof}
By induction. For the trivial harpoons $H_{i,i}$, the maximum independent set has weight 0, because this set is empty.
The correctness of the theorem follows directly from the previous lemma and the next implications: For $i\neq j$,
\begin{align*}
i < j &\Longrightarrow S(i,j) = \max\left\{S(i,j-1),\max_{k\in \hat{B}^{j}_h \cap H_{i,j}}S_k(i,j)\right\}.\\
j < i &\Longrightarrow S(i,j) = \max\left\{S(i,f(j)),\max_{k\in \hat{B}^{j}_v \cap H_{i,j}}S_k(i,j)\right\}.
\end{align*}
Indeed, assume that $i<j$ (the case $i>j$ is analogous). Let $\mathcal{S}$ be the \MIS
corresponding to $S(i,j)$, and let $m \in \mathcal{S}$ be minimal with respect to domination.
If $m$ is in $H_{i,j-1}$ then $ S(i,j) = S(i,j-1)$. Otherwise, $m$ is in $\hat{B}^{j}_h$ and since $\mathcal{S}\setminus\{m\}$ does not contain
rectangles dominated by $m$,  $S(i,j)= S_{m}(i,j)$.
\end{proof}

\subsection{An approximation for diagonal-intersecting families}
We use the previous algorithm to get a 2-approximation for diagonal-intersecting rectangle families.
This improves upon the 6-approximation (which is only for the unweighted case) of Chepoi and Felsner~\cite{ChepoyF13}.

\begin{Theo}
There exists a 2-approximation polynomial algorithm for \WMIS on diagonal-intersecting rectangle families.
\end{Theo}

\begin{proof}
Divide $\RR$ into two subsets: the rectangle that intersect the diagonal on their upper side, and the ones that don't. It is easy to see that every
rectangle in the second subset intersect the diagonal on its left side. Using symmetry, the left side case is equivalent to the upper side case.
Therefore we can compute in polynomial time a \WMIS in each subset. We output the heaviest one. Its weight is at least half of $\wmis(\RR)$. This
algorithm gives a 2-approximation
\end{proof}

\section{Duality gap and other approximation algorithms}\label{duality}
In this section we explore the duality gap, that is, the largest possible ratio between $\mhs$ and $\mis$, on some of the rectangle classes defined
before.

\begin{Theo} \label{teo:dual} The duality gap for diagonal-touching rectangle families is between 3/2 and~2.
For diagonal-lower-intersecting families it is between 2 and 3, and for diagonal-intersecting families it is between 2 and 4.
\end{Theo}

We will prove the upper bounds and the lower bounds separately.

\begin{proof}[Proof of the upper bounds in Theorem~\ref{teo:dual}]
Let $\RR$ be a rectangle family in the plane, that can be in one of the three classes described on the theorem.
In the case which $\RR$ is diagonal-lower-intersecting we first replace each rectangle $r\in \RR$ by the minimal one containing the region of $r$
that is below the diagonal.
The modified family has the same intersection graph as before, but it is diagonal-corner-separated. In particular, the region of each rectangle that
is above
the diagonal is a triangle or a single point.

We use $\RR_x$ and $\RR_y$ to denote the projections of the rectangles in $\RR$ on the $x$-axis and $y$-axis respectively.
Both $\RR_x$ and $\RR_y$ can be regarded as intervals, and so we can compute in polynomial time the minimum hitting sets, $P_x$ and $P_y$, and the
maximum independent sets, $\II_x$ and $\II_y$, of $\RR_x$ and $\RR_y$ respectively. Since interval graphs are perfect,
$|P_x|=|\II_x| \text{ and } |P_y|=|\II_y|$.

Furthermore, since rectangles with disjoint projections over the $x$-axis (resp. over the $y$-axis) are disjoint, we also have
\begin{align*}
\mis(\RR) \geq \max\{|\II_x|,|\II_y|\} = \max\{|P_x|,|P_y|\}.
\end{align*}
Observe that the collection $\mathcal{P}=P_x\times P_y \subset \Real^2$ hits every rectangle of $\RR$. From here we get the (trivial) bound
$\mhs(\RR)\leq |\mathcal{P}|\leq \mis(\RR)^2$ which holds for every rectangle family. When $\RR$ is in one of the classes studied in this paper, we
can improve the bound.

Let $\mathcal{P}^-$ and $\mathcal{P}^+$ be the sets of points in $\mathcal{P}$ that are below or above the diagonal, respectively. Consider the
following subsets of $\mathcal{P}$:
\begin{align*}
\mathcal{F}^- &= \{p \in \mathcal{P}^-\colon \nexists q \in \mathcal{P}^-\setminus\{p\}, p_x< q_x \text{ and } p_y< q_y\}. \\
\mathcal{F}^+ &= \{p \in \mathcal{P}^+\colon \nexists q \in \mathcal{P}^+\setminus\{p\}, q_x< p_x \text{ and } q_y< p_y\}. \\
\mathcal{F}^* &= \{p \in \mathcal{P}^+\colon \nexists q \in \mathcal{P}^+\setminus\{p\}, q_x\leq p_x \text{ and } q_y\leq p_y\}.
\end{align*}

The set $\mathcal{F}^-$ (resp. $\mathcal{F}^+$) forms the closest ``staircase'' to the diagonal that is below (resp. above) it. The set
$\mathcal{F}^*$ corresponds to the lower-left bending points of the staircase defined by $\mathcal{F}^+$. See Figure \ref{fig:F_sets}).

\begin{figure}[h]
\centering
\includegraphics[scale=1.3]{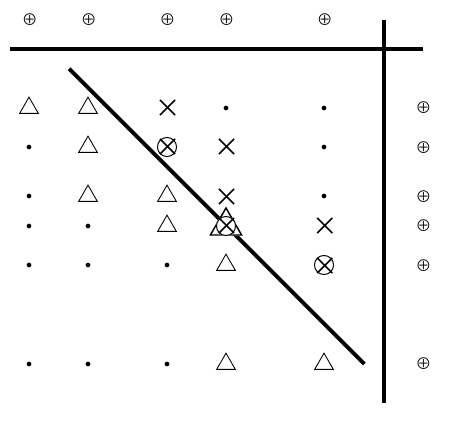}
\caption{\small We do not represent the rectangles but $\mathcal{P}_x$ and $\mathcal{P}_y$ (the plus into circles, along the axis). The points of $\mathcal{P}\setminus (\mathcal{F}^+ \cup \mathcal{F}^-)$ are the dots, $\mathcal{F}^-$ corresponds to the triangles, $\mathcal{F}^+$ corresponds to the 'x'-s, and $\mathcal{F}^*$ corresponds to the circles. Remark that a point can be in several sets.}\label{fig:F_sets}
\end{figure}

From here, it is easy to see that
\begin{align*}
\max\{|\mathcal{F}^-|,|\mathcal{F}^+| \} &\leq |P_x|+|P_y|-1 \leq 2 \mis(\RR) - 1.\\
|\mathcal{F}^*| &\leq \max\{|P_x|,|P_y|\} \leq \mis(\RR).
\end{align*}

If $r\in \RR$ is hit by a point of $\mathcal{P}^-$, let $p_1(r)$ be the point of $\mathcal{P}^-\cap r$ closest to the diagonal (in $\ell_1$-distance).
Since $r$ intersects the diagonal, and the points of $\mathcal{P}$ form a grid, we conclude that $p_1(r)\in\mathcal{F}^-$.
Similarly, if $r\in \RR$ is hit by a point of $\mathcal{P}^+$, let $p_2(r)$ be the point of $\mathcal{P}^+\cap r$ closest to the diagonal. Since
$r$ intersects the diagonal, we conclude that $p_2(r)\in\mathcal{F}^+$. Furthermore, if the region
of $r$ that is above the diagonal is a triangle, then $p_2(r) \in \mathcal{F}^*$.

If $\RR$ is diagonal-touching, then every rectangle is hit by a point of $\mathcal{F}^-$, and so $\mhs(\RR) \leq |\mathcal{F}^-| \leq 2 \mis(\RR) -
1$.
If $\RR$ is diagonal-lower-intersecting (and, after the modification discussed at the beginning of this proof, diagonal-corner-separated), then every
rectangle
is hit by a point of $\mathcal{F}^-\cup \mathcal{F}^*$, and so $\mhs(\RR) \leq |\mathcal{F}^-|+ |\mathcal{F}^*| \leq 3 \mis(\RR) - 1$.
Finally, if $\RR$ is diagonal-intersecting, then every rectangle is hit by a point of $\mathcal{F}^-\cup \mathcal{F}^+$, and so $\mhs(\RR) \leq
|\mathcal{F}^-|+|\mathcal{F}^+| \leq 4 \mis(\RR) - 2$.\end{proof}

\begin{proof}[Proof of the lower bounds of Theorem~\ref{teo:dual}]
The lower bound of 3/2 is achieved by any family $\RR$ whose intersection graph $G$ is a 5-cycle. It is easy to see that $\RR$ can be realized as a
diagonal-touching family, that $\mis(\RR)=2$ and $\mhs(\RR)=3$, and so the claim holds.

The lower bound of $2$ for diagonal-lower-intersecting and diagonal-intersecting families is asymptotically attained by a sequence of rectangle
families $\{\RR_k\}_{k\in \ZZ^+}$. We will describe the sequence in terms of infinite rectangles which intersect the diagonal, but it is easy to transform each $\RR_k$ into a family of
finite ones by considering a big bounding box.

For $i \in \ZZ^+$, define the $i$-th layer as $\mathcal{L}_i = \{U(i), D(i), L(i), R(i)\}$, and for $k \in \ZZ^+$, define the $k$-th instance as
$\RR_k = \bigcup_{i=1}^k \mathcal{L}_i$, where:
\begin{align*}
U(i)&=[2i,2i+1]\times [-(2i+\tfrac13),+\infty), &D(i)&=[2i+\tfrac23,2i+\tfrac53]\times (-\infty,-2i]\\
L(i)&=(-\infty, 2i+\tfrac13]\times [-2i-1,-2i], &R(i)&=[2i, \infty)\times [-(2i+\tfrac53),-(2i+\tfrac23)].
\end{align*}

\begin{figure}[h]
\centering
\includegraphics[scale=0.8]{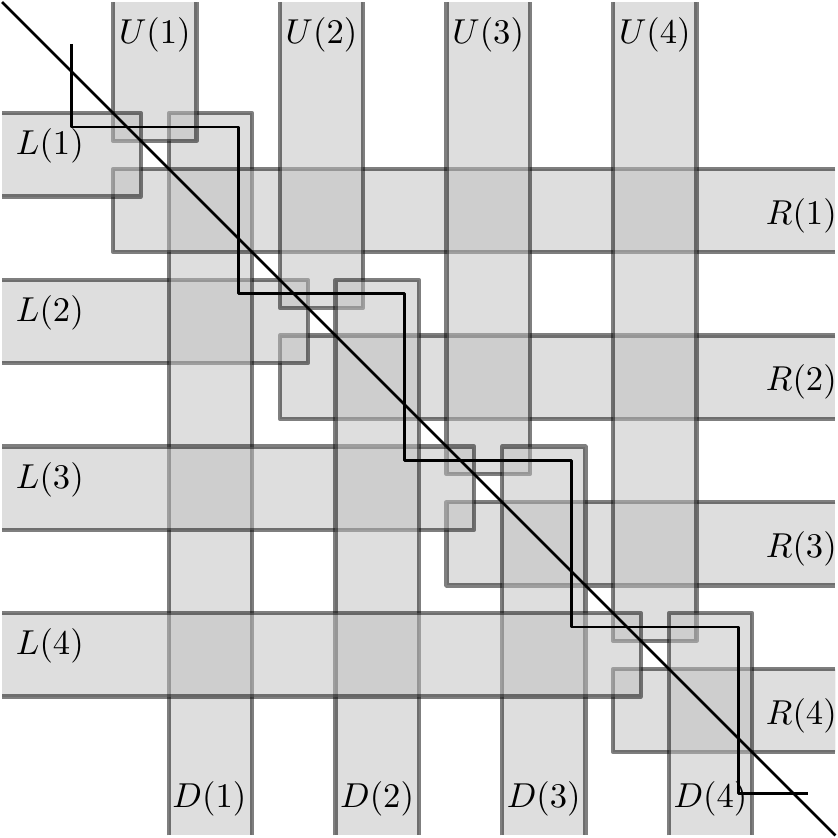}
\caption{\small The family $\RR_4$. The diagonal line shows this family is diagonal-intersecting. The staircase line shows that it is actually
lower-diagonal-intersecting.}\label{fig:gap2}
\end{figure}

Consider the instance $\RR_k$ depicted in Figure~\ref{fig:gap2} with $k$ layers
of rectangles. $\RR_k$ can be easily transformed into a diagonal-lower-intersecting
family by ``straightening'' the staircase curve shown in the figure
without changing its intersection graph. Let $I$ be a maximum independent set of rectangles
in that instance. It is immediately clear that a minimum hitting set
has size $2k$ since no point in the plane can hit more that two
rectangles.

Let us prove that the size of a maximum independent set is at most
$k+2$, amounting to conclude that the ratio is arbitrarily close to 2.
To this end, we let $i_D=\min\{i: D(i)\in I\}$ and  $i_R=\min\{i: R(i)\in I\}$, and if no $D(i)\in I$ or no $R(i)\in I$, we let $i_D=k+1$ or
$i_R=k+1$,
respectively. When $i_D=i_R=k+1$, it is immediate that $|I|\leq k$. Assume then, without loss of generality, that $i_D<i_R$.

Since for $i=1,\ldots, i_D-1$ the set $I$ neither contains rectangle
$D(i)$ nor $R(i)$, we have that $I$ contains at most one rectangle on each of these layers.
It follows that $|I \cap \cup_{i=1}^{i_D-1} \mathcal{L}_i|\le i_D-1$. Similarly, for $i=i_D+1,\ldots, i_R-1$ the
set $I$ neither contains rectangle $L(i)$ nor $R(i)$, thus $|I \cap
\cup_{i=i_D+1}^{i_R-1} \mathcal{L}_i|\le i_R-i_D-1$.  Finally, we have that for
$i=i_R+1,\ldots, k$ the set $I$ neither contains rectangle $L(i)$ nor
$U(i)$, and on layer $i_R$, $I$ contains at most 2 rectangles; thus $|I \cap \cup_{i=i_R}^{k} \mathcal{L}_i|\le k- i_R+2$. To
conclude, note that $I$ may contain at most 2 rectangles of layer $i_D$, then
\begin{equation*}
|I| = \sum_{i=1}^k |I \cap \mathcal{L}_i| \le i_D - 1+i_R - i_D - 1 + k -  i_R +2 +2 = k+2.\qedhere
\end{equation*}
\end{proof}
\begin{Coro} There is a simple 2-approximation polynomial time algorithm for \MHS on diagonal-touching families, a 3-approximation for \MHS on
diagonal-lower-intersecting families, and a 4-approximation polynomial time algorithm for \MHS on diagonal-intersecting families.
\end{Coro}
\begin{proof}
The algorithm consists in computing and returning $\mathcal{F}^-$ for the first case, $\mathcal{F}^-\cup \mathcal{F}^*$ for the second one, and
$\mathcal{F}^-\cup \mathcal{F}^+ $ for the third one.
\end{proof}

\section{NP-hardness of \MIS for diagonal-intersecting families}\label{sec:hardness}

In this section we prove the following theorem.
It is worth noting that the class of rectangles it refers to is not the class of diagonal-touching rectangles: some of the rectangles may touch the diagonal on its lower-left corner while others may touch it on its upper-right one.

\begin{Theo}\label{theo:MISR-hardness}
The \MIS problem is NP-hard on diagonal-intersecting families of rectangles, even if the diagonal intersects each rectangle on a corner.
\end{Theo}

\begin{proof}
We use a reduction from the \pb{Planar 3-SAT} which is NP-complete~\cite{Lichtenstein1982}.
The input of the \pb{Planar 3-SAT}
consists of a Boolean formula $\varphi$ in 3-CNF whose associated
graph is planar, and the formula is accepted if and only if there
exists an assignment to its variables such that in each clause
at least one literal is satisfied.
Let $\varphi$ be a planar 3-SAT formula.
The (planar) graph associated with $\varphi$ can
be represented in the plane as in Figure~\ref{fig:planar3SAT-sample},
where all variables lie on an horizontal line,
and all clauses are represented by {\em non-intersecting} three-legged combs~\cite{Knuth1992}.
We identify each clause with its corresponding comb, and vice versa.
Using this embedding as base, which can be constructed in a grid of polynomial size~\cite{Knuth1992},
we construct a set $\mathcal{R}$ of rectangles intersecting the diagonal $D$,
such that there exists in $\mathcal{R}$ an independent set of some given number of rectangles
if and only if $\varphi$ is accepted. Such a construction of $\mathcal{R}$
follows ideas of Caraballo et al.~\cite{caraballo2013}.

\begin{figure}[h]
	\centering	
	\includegraphics[scale=1]{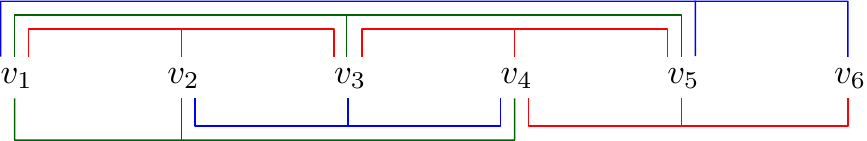}		
	\caption{\small{Planar embedding of $\varphi=(v_1 \lor \overline{v_2} \lor v_3)
	\land (v_3 \lor \overline{v_4} \lor \overline{v_5})\land(\overline{v_1} \lor \overline{v_3} \lor v_5)
	\land(v_1 \lor \overline{v_2} \lor v_4)\land(\overline{v_2} \lor \overline{v_3} \lor \overline{v_4})
	\land(\overline{v_4} \lor v_5 \lor \overline{v_6})\land(\overline{v_1} \lor v_5 \lor v_6)$.}}
\label{fig:planar3SAT-sample}
\end{figure}

Let $\varphi$ be an instance of the \pb{Planar 3-SAT},
with $n$ variables and $m$ clauses, and let $E_0$ denote
the above embedding of $\varphi$.
For any variable $v$,
let $d(v)$ denote the number of clauses in which $v$ appears.
We assume that every variable appears in each clause at most once.
Given any clause $C$ with variables $u$, $v$, and $w$, such that
$u$, $v$, and $w$ appear in this order from left to right in the embedding $E_0$,
we say that $u$ is the {\em left} variable of $C$,
and that $w$ is the {\em right} variable.
Given $E_0$, using simple geometric transformations we can obtain the next slightly different
planar embedding $E_1$ of $\varphi$. Essentially, $E_1$ can be obtained
from $E_0$ by arranging the variables in the diagonal $D$ and extending the combs.
Such an embedding $E_1$ has the next further properties (see Figure~\ref{fig:np-hard-1}):

\begin{figure}[ht]
	\centering
	\includegraphics[scale=0.5,page=1]{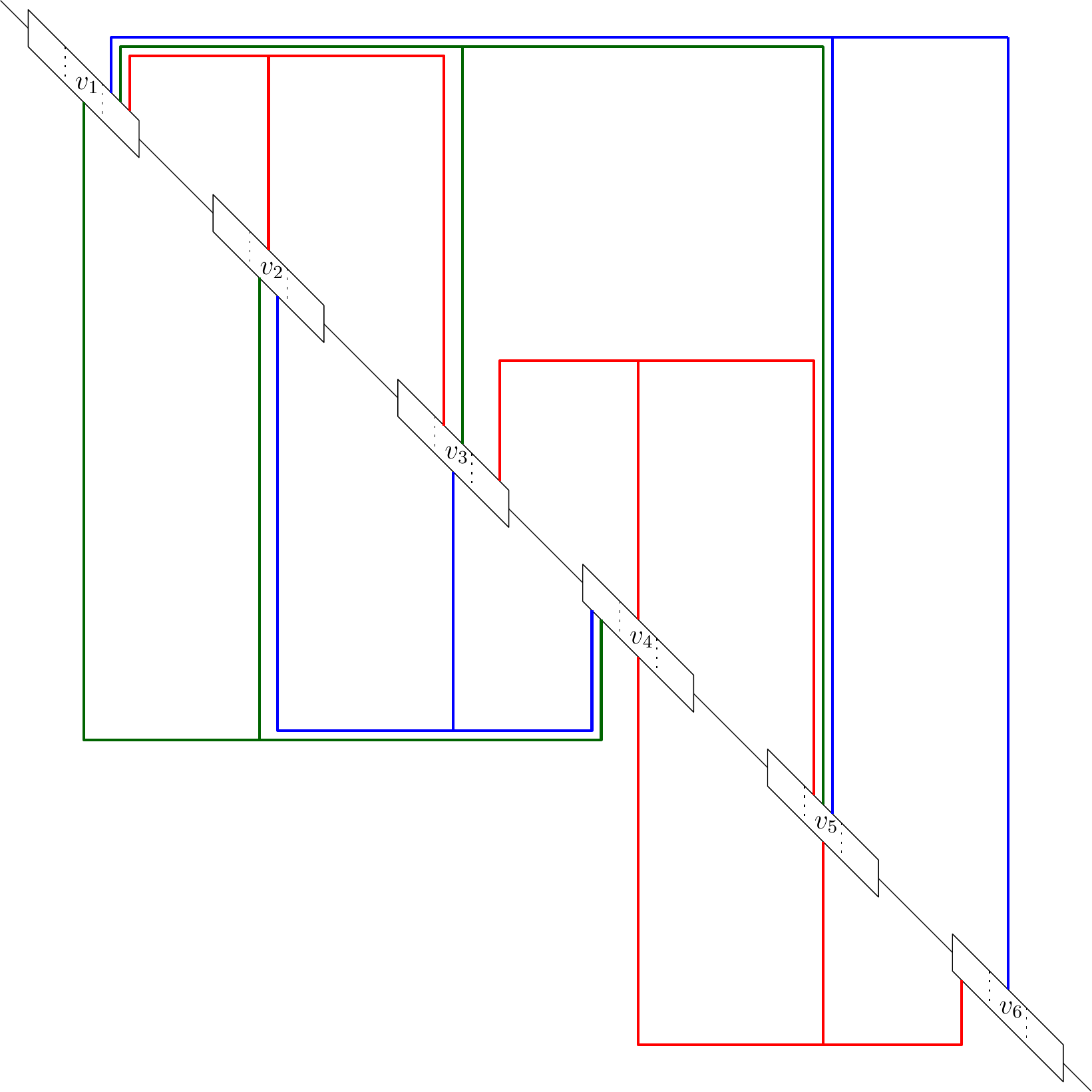}
	\caption{\small{The embedding $E_1$ obtained (essentially) from $E_0$ by arranging the variables
	as segments in the diagonal $D$ and extending the combs. For clarity, each variable segment is
	represented by a parallelogram.}}
	\label{fig:np-hard-1}
\end{figure}

\begin{itemize}[leftmargin=17pt]
\item[(1)] Each variable $v$ is represented by a segment $S_v\subset D$, divided
into three equal parts: $S^{\ell}_v$ is the left part, $S^m_v$ the middle one, and
$S^r_v$ the right one. The segments $S_v$'s are pairwise disjoint and equally spaced
in $D$, and appear in $D$ from left to right as the $v$'s appear in $E_0$. Let $\delta$
denote the vertical gap between successive segments $S_v$'s.

\item[(2)] If the variable $v$ is the left variable of a clause $C$ above $D$,
then the left leg of $C$ (i.e.\ the one corresponding to $v$)
contacts the interior of $S^r_v$. Otherwise, it contacts the interior of $S^m_v$.

\item[(3)] If the variable $v$ is the right variable of a clause $C$ below $D$,
then the right leg $C$ (i.e.\ the one corresponding to $v$)
contacts the interior of $S^{\ell}_v$. Otherwise, it contacts the interior of $S^m_v$.
\end{itemize}

The above properties (1)-(3) of $E_1$ allows us to obtain the next embedding $E_2$ of $\varphi$
(refer to Figure~\ref{fig:np-hard-2}).
Let $v$ be any variable. 
Let $C_1,C_2,\ldots,C_k$ be the clauses above the diagonal having $v$ as left variable, sorted
according to the {left-to-right} order of the contact points of their left legs with $S^r_v$.
Let $s_1,s_2,\ldots,s_k$ denote the horizontal segments of $C_1,C_2,\ldots,C_k$, respectively.
Assume w.l.o.g.\ that $s_1,s_2,\ldots,s_k$ are equally
spaced at a distance less than $\delta/2k$. Push downwards simultaneously $s_1,s_2,\ldots,s_k$ to
modify $C_1,C_2,\ldots,C_k$
so that $s_1$ is now below $S_v$, the vertical gap between $s_k$ and $S_v$ is less than $\delta/2$,
and the left legs of $C_1,C_2,\ldots,C_k$ are inverted and make contact with $S^r_v$ from below.
Further modifying $C_1,C_2,\ldots,C_k$ by inverting the left-to-right order of the contact
points of $C_1,C_2,\ldots,C_k$ with $S^r_v$,
$C_1,C_2,\ldots,C_k$ become pairwise disjoint.
Proceed similarly (and symmetrically) with the clauses below the diagonal having $v$ as right variable.
It can be verified that this new embedding $E_2$ has no crossings
among the combs and variable segments.

\begin{figure}[t]
	\centering
	\includegraphics[scale=0.5,page=2]{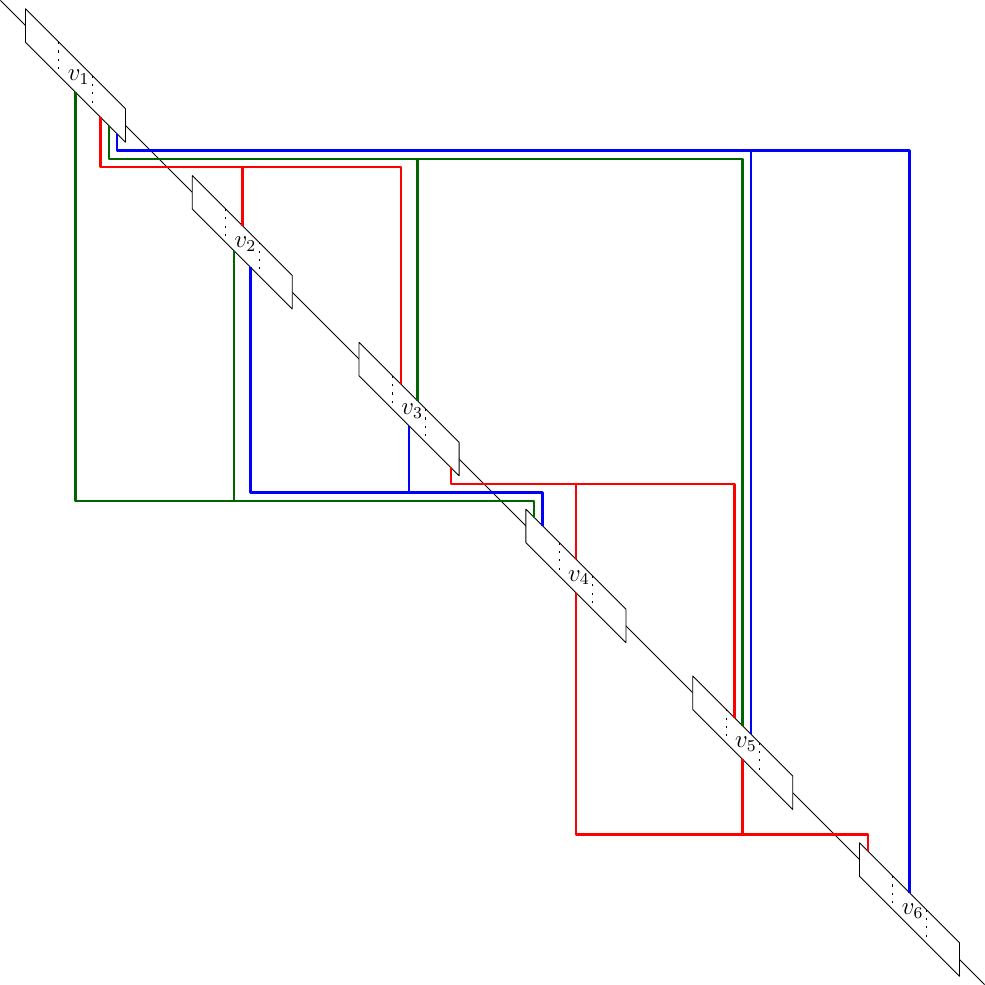}
	\caption{\small{The embedding $E_2$ obtained from $E_1$ by modifying the combs
	and inverting one leg.}}
	\label{fig:np-hard-2}
\end{figure}

Using the embedding $E_2$, we construct a set $\mathcal{R}$ or rectangles via variable gadgets and
clause gadgets.
\smallskip
\paragraph*{\textbf{Variable gadgets:}} For each variable $v$, the segment $S_v$ is replaced by
a {\em necklace} $Q_v$ of $12\cdot d(v)+2$ squares (their intersection graph is a cycle), so that each square intersects the diagonal~$D$ on a corner and only consecutive squares pairwise intersect
(see Figure~\ref{fig:np-hard-3}). We number these squares consecutively in clockwise order, starting from the topmost one which is numbered 0.
Let $Q^1_v\subset Q_v$ be the first top-down $2\cdot d(v)$ squares above $D$,
$Q^2_v\subset Q_v$ the second top-down $2\cdot d(v)$ squares above $D$,
$Q^3_v\subset Q_v$ the first bottom-up $2\cdot d(v)$ squares below $D$,
and $Q^4_v\subset Q_v$ the second bottom-up $2\cdot d(v)$ squares below $D$.
Since any clause can contact: $S^{\ell}_v$ from above, $S^{r}_v$ from below, and $S^m_v$ from
either above or below, we identify $S^{\ell}_v$ with $Q^1_v$,
$S^m_v$ with $Q^2_v\cup Q^4_v$, and $S^{r}_v$ with $Q^3_v$.

\begin{figure}[ht]
	\centering
	\includegraphics[scale=0.6,page=3]{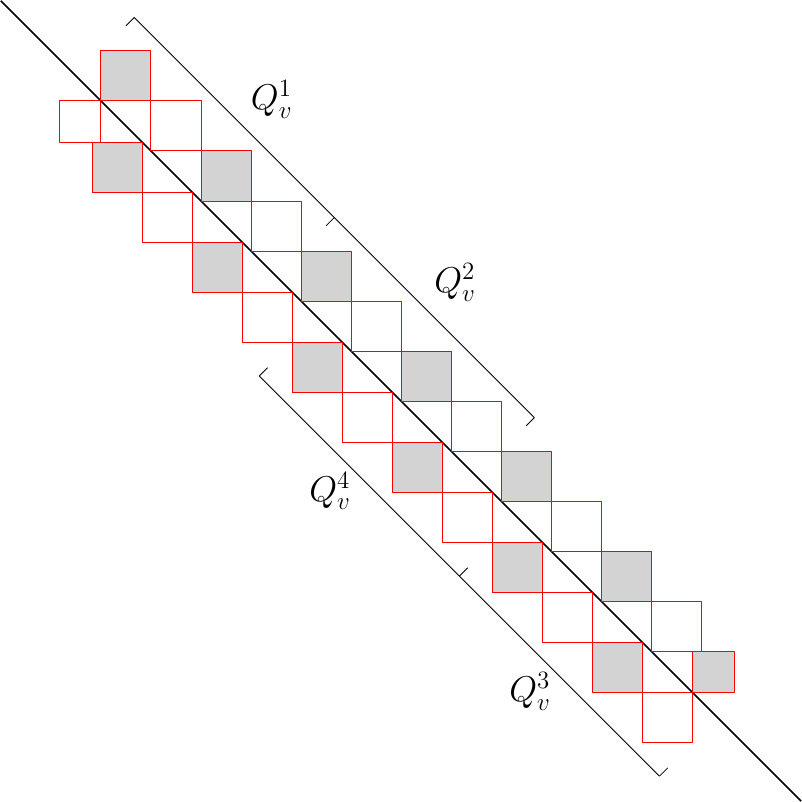}
	\caption{\small{The variable gadget for the variable $v$. The even numbered
	squares are shaded.}}
	\label{fig:np-hard-3}
\end{figure}

\smallskip\paragraph*{\textbf{Clause gadgets:}} Let $C$ be a clause with variables $u$, $v$, and $w$, appearing
in this order from left to right in $E_0$.
We represent $C$ by the set $Q_C$ of nine {\em thin} rectangles, three vertical and six horizontal,
as in Figure~\ref{fig:np-hard-4}.
The vertical rectangles of $Q_C$ represent the three legs of $C$, and
for $z=u,v,w$ the vertical rectangle corresponding to $z$ intersects a unique rectangle $R_z$ of $Q_v$,
and $D$ as well, so that $R_z$ is even numbered if and only if $z$ appears as positive in $C$.
Furthermore, if $C$ is above $D$ in $E_1$ then $R_u\in Q^3_u$, $R_v\in Q^2_v$, and $R_w\in Q^2_w$.
Otherwise, if $C$ is below $D$ in $E_1$ then $R_u\in Q^4_u$, $R_v\in Q^4_v$, and $R_w\in Q^1_w$.
Observe that since for every variable $z$ each of the sets $Q^1_z$, $Q^2_z$, $Q^3_z$, $Q^4_z$ contains
$2\cdot d(z)$ squares, we can guarantee that each square of the variable gadgets
is intersected by at most one vertical rectangle of the clause gadgets.

\begin{figure}[h!]
	\centering
	\includegraphics[scale=0.48,page=4]{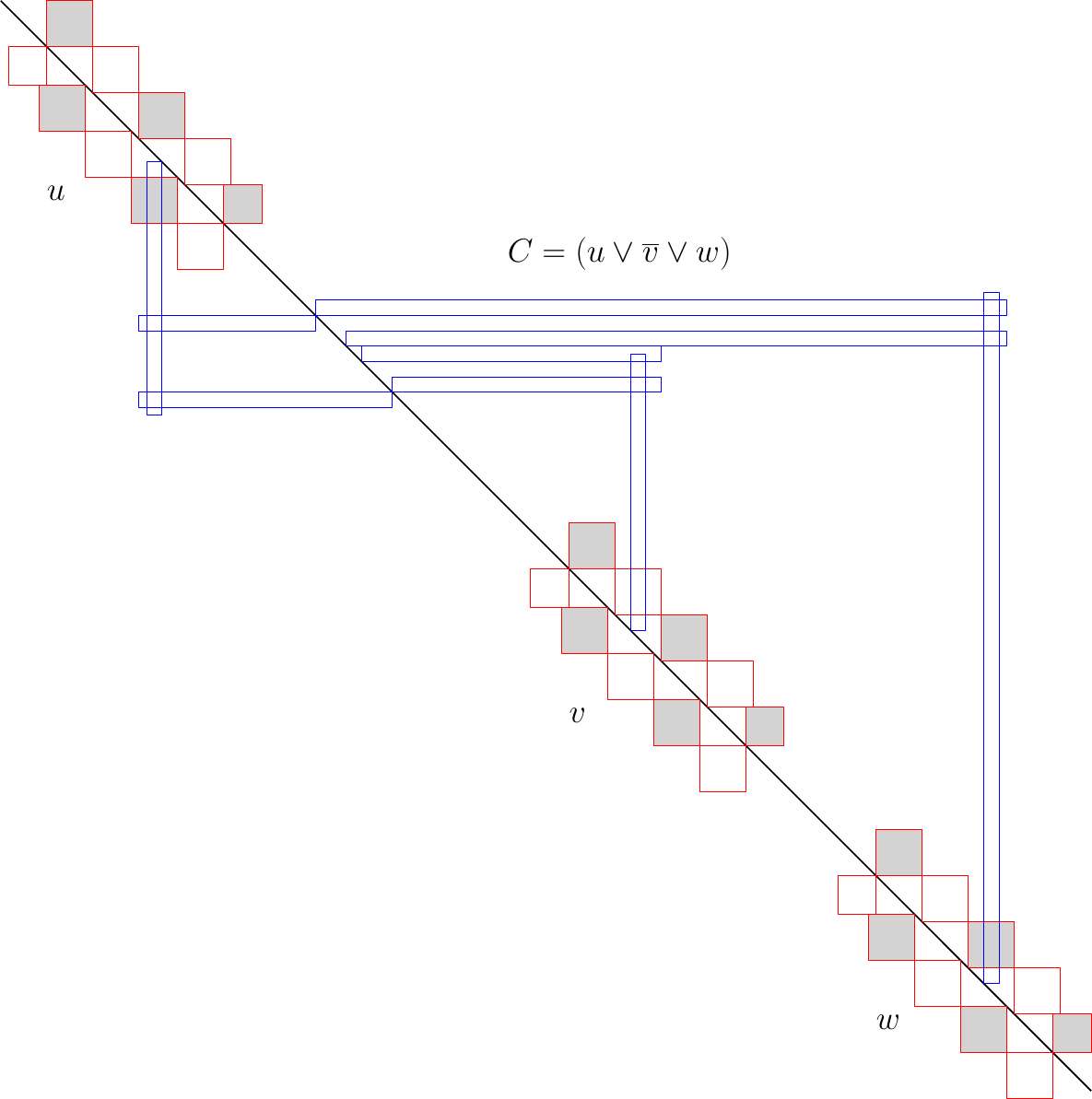}
	\caption{\small{Clause gadget for the clause $C=(u\lor\overline{v}\lor w)$.
	The variable $u$ is positive in $C$ and the vertical rectangle of $C$ corresponding
	to $u$ intersects $D$ and an even numbered square of $Q^3_u$.
	The variable $v$ appears negative and the vertical rectangle of $C$ corresponding
	to $v$ intersects $D$ and an odd numbered square of~$Q^4_v$.
	The variable $w$ appears positive and the vertical rectangle of $C$ corresponding
	to $w$ intersects $D$ and an even numbered square of~$Q^3_w$.
	}}
	\label{fig:np-hard-4}
\end{figure}

\smallskip\paragraph*{\textbf{Reduction:}}
Observe that in each variable $v$, the set $Q_v$ has exactly two maximum independent sets
of rectangles of size $6\cdot d(v) + 1$: the set $Q_{v,0}\subset Q_v$ of the even-numbered squares and the set
$Q_{v,1}\subset Q_v$ of the the odd-numbered squares.
We consider that $v=1$ if we select $Q_{v,1}$ as a maximum
independent set of rectangles of $Q_v$, and consider $v=0$ if $Q_{v,0}$ is selected.
Observe that the next statements are satisfied:
\begin{itemize}[leftmargin=17pt]
\item[(1)] if $v=1$ then the vertical rectangles of the clause gadgets in which
$v$ appears as positive, together with $Q_{v,1}$, form an independent set of rectangles.

\item[(2)] if $v=0$ then the vertical rectangles of the clause gadgets in which
$v$ appears as negative, together with $Q_{v,0}$, form an independent set of rectangles.

\item[(3)] For each clause $C$, any maximum independent set of rectangles
of $Q_C$ has size 4, and among its elements there must be a vertical  rectangle.

\item[(4)] For each clause $C$ and each variable $v$ in $C$: there exists an independent
set of size $(6\cdot d(v)+1)+4$ in $Q_v \cup Q_C$, such that the vertical
rectangle of $C$ corresponding to $v$ is selected, if and only if either $v$ appears as positive in $C$
and $Q_{v,1}$ is selected or $v$ appears as negative in $C$
and $Q_{v,0}$ is selected.
\end{itemize}
Let $\mathcal{R}$ be the set of the rectangles of all variable gadgets and clause gadgets.
From the above observations, we claim that $\varphi$ can be accepted if and only if
$\mathcal{R}$ has an independent set of exactly $\sum_{v}(6\cdot d(v)+1)+4m$ rectangles.

Therefore, the MIS problem is NP-hard on diagonal-intersecting families of rectangles (even if the diagonal intersects each rectangle on a corner) since $\mathcal{R}$ is a family of such rectangles.
\end{proof}

\section{The duality gap of general rectangle families is $O((\log \log (\mis))^2)$}
\label{general-duality}
In this section, we prove that the duality gap of general rectangle families is $O((\log \log (\mis))^2)$. This observation is a simple application of the results in \cite{Chalermsook09} and \cite{Aronov10} however, as far as we know, it is not a known result.

\begin{Theo}\label{teo:loglog} For every rectangle family $\RR$
\begin{equation*}
\frac{\mhs(\RR)}{\mis(\RR)} \leq O\bigl((\log \log \mis(\RR))^2\bigr).
\end{equation*}
\end{Theo}

In order to prove this bound, we need to briefly recall the natural linear programming formulations for \MIS and \MHS. A point $p$ in the plane is
called a \emph{witness} of a maximal clique of the intersection graph of $\RR$ if it hits this clique. Given any (possibly infinite) hitting set $H
\subseteq \Real^2$ of containing the witness points of all maximal cliques in $\RR$, define the following polytopes:
\begin{align*}
\text{Pol}_H(\RR) &= \big\{x \in \Real^\RR\colon \sum_{r\in \RR:\ p \in r}x_r \leq 1 \text{ for all $p \in H$}, x \geq 0\big\},\\
\text{Dual}_H(\RR) &= \big\{y \in \Real^H\colon \sum_{p \in H\cap r}y_p \geq 1 \text{ for all $r \in \RR$}, y \geq 0\big\}.
 \end{align*}
and the following dual linear programs:
\begin{align*}
  \LP_H(\RR) &= \max\bigg\{ \sum_{r \in \RR}x_r\colon x \in \text{Pol}_H(\RR)\bigg \},\\
  \LP'_H(\RR) &= \min\bigg\{ \sum_{p \in H}y_p\colon y \in \text{Dual}_H(\RR)\bigg \}.
\end{align*}
It is easy to see that the previous linear programs are relaxations of \MIS and \MHS respectively, therefore
\begin{equation}\label{eqn:dual}
 \mis(\RR) \leq \LP_H(\RR) = \LP'_H(\RR) \leq \mhs(\RR);
 \end{equation}
furthermore the value $\LP_H(\RR)$ does not depend on the hitting set $H$ chosen, and so, we can drop the subindex $H$ in \eqref{eqn:dual}.

Chalermsook and Chuzhoy~\cite{Chalermsook09} showed that for any family of rectangles $\RR$ having their corners in the grid $[t]^2=[t]\times[t]$, it
is possible to find an independent set $\mathcal{Q}$ with
\begin{equation}\label{eq:loglog-MIS}
\mis(\RR) \leq \LP_{[t]^2}(\RR) \leq |\mathcal{Q}|\cdot O(\log \log (t) ),
\end{equation}

On the other hand, Aronov et al.~\cite{Aronov10} have shown the existence of $O(\frac{1}{\varepsilon}\log\log\frac{1}{\varepsilon})$-nets for families of axis-parallel rectangles, concluding that for every family $\RR$ of rectangles, there exist a polynomial time computable hitting set $P$, with
\begin{equation}\label{eq:loglog-MHS}
|P| \leq O(\LP(\RR) \log \log (\LP(\RR))).
\end{equation}

To prove Theorem~\ref{teo:loglog}, we require the following lemma.

\begin{Lem}\label{lem:lemfornewbound}
For every family of rectangles $\RR$, with $\alpha := \mis(\RR)$, there is another family $\RR'$ of rectangles with corners in the grid $[\alpha]^2$
such that
\begin{align}\label{eq:lemma-gridify}
  \mis(\RR') \leq \mis(\RR) \leq \LP(\RR)  \leq 9\LP(\RR').
\end{align}
\end{Lem}
\begin{proof} Let $\RR_x$ (resp. $\RR_y$) be the family of intervals obtained by projecting $\RR$ on the $x$-axis (resp. $y$-axis), and let $P_x$ (resp.~$P_y$) be minimum hitting sets for $\RR_x $ (resp.~$\RR_y$). Similar to the proof of the upper bounds of Theorem~\ref{teo:dual}, we have
\begin{align*}
\max\{|P_x|,|P_y|\} = \max\{\mis(\RR_x),\mis(\RR_y)\} \leq \mis(\RR) = \alpha.
\end{align*}

Consider the grid $P_x \times P_y$ of size at most $\alpha \times \alpha$. By translating and piece-wise scaling the plane, we can identify $P_x$ with
the set $\{(i,0)\colon 1\leq i \leq |P_x|\}$ and $P_y$ with the set $\{(0,j)\colon 1 \leq j \leq |P_y|\}$ without changing the intersection graph
associated with $\RR$. Thus, we can identify the grid $P_x \times P_y$ with a subgrid of $[\alpha]\times[\alpha]$. Note that this grid is itself, a
hitting set of $\RR$.

Furthermore, consider the family $\widetilde{\RR} =\{R \cap [1,\alpha] \times [1,\alpha]\colon R \in \RR\} $. This is, $\widetilde\RR$ is obtained by trimming the rectangles to the rectangular region $[1,\alpha]\times[1,\alpha]$. It is easy to see that this operation does not change the intersection graph of the family either. So, for our purposes, we will assume w.l.o.g. that $\RR = \widetilde\RR$.

Let $\RR'$ be the family of rectangles obtained by replacing each rectangle $r$ of $\RR$ by the minimal possible rectangle in the plane containing $r$
and having all its corners in the grid $[\alpha]\times[\alpha]$. This is, we replace the rectangle $r$ defined by $\ell^r$ and $u^r$ by the rectangle
defined by $\tilde{\ell}^r=(\lfloor \ell^r_x \rfloor, \lfloor \ell^r_y \rfloor)$ and $\tilde{u}^r=(\lceil u^r_x \rceil, \lceil u^r_y \rceil)$, where
$\lfloor\cdot\rfloor$ and $\lceil\cdot\rceil$ are the floor and ceiling functions, respectively.

The first inequality of \eqref{eq:lemma-gridify} follows since any independent set of $\RR'$ induces an independent set of $\RR$ of the same size. The second inequality follows from \eqref{eqn:dual}. The only non-trivial inequality is the last one.

Since $[\alpha]^2$ is a hitting set for $\RR$ and $\RR'$, $\LP_{[\alpha]^2}(\RR)=\LP(\RR)$ and $\LP_{[\alpha]^2}(\RR')=\LP(\RR)$.
Consider a fractional optimal solution $y'$ for $\LP'_{[\alpha]^2}(\RR')$ and recall that the support of $y'$ is contained in $[\alpha]^2$. Observe that if $p$ is a point in the support of~$y$ that fractionally hits some grown rectangle $r_+$, then either $p$, one of its 4 immediate neighbors in the grid or one of its 4 diagonal neighbors in the grid will hit the original rectangle $r$. Define $y$ as
\begin{equation}
y_q = y'_q + \sum_{\substack{p \in [\alpha]^2\colon p \text{ immediate or }\\ \text{diagonal neighbor of } q}} y'_p, \text{ for all $q\in [\alpha]^2$.}
\end{equation}

By the previous observation, $y$ is a fractional feasible solution for the dual of $\LP_{[\alpha]^2}(\RR)$, and by definition, its value is at most $9$ times the value of $y'$.\qedhere
\end{proof}

Now we are ready to prove Theorem~\ref{teo:loglog}.
\begin{proof}[Proof of Theorem~\ref{teo:loglog}]
Let $\RR$ be a family of rectangles with $\mis(\RR)=\alpha$ and $\RR'$ the family guaranteed by Lemma~\ref{lem:lemfornewbound}
Then, by combining \eqref{eq:loglog-MHS} and \eqref{eq:loglog-MIS}, we have:
\begin{align*}
  \mhs(\RR) &\leq O(\LP(\RR)\log\log(\LP(\RR))) \leq O(\LP(\RR')\log\log(\LP(\RR')))\\
              &= O(\LP_{[\alpha]^2}(\RR')\log\log(\LP_{[\alpha]^2}(\RR')))\\
              &\leq O(\alpha \log\log(\alpha) \log\log(\alpha \log\log (\alpha))) = O(\alpha (\log\log(\alpha))^2).\qedhere
\end{align*}
\end{proof}

\section{Graph classes inclussions}
\label{graph-classes}
\noindent \textbf{Lemma 1.}
Let $\mathcal{G}_{\text{int}}=\{\II(\RR)\colon \RR \text{ is diagonal-intersecting}\}$ be the class of intersection graphs arising from
diagonal-intersecting families of rectangles. Let also $\mathcal{G}_{\text{low-int}}$, $\mathcal{G}_{\text{split}}$,
$\mathcal{G}_{\text{c-sep}}$ and $\mathcal{G}_{\text{touch}}$ be the classes arising from diagonal-lower-intersecting,
diagonal-splitting, diagonal-corner-separated, and diagonal-touching families of rectangles, respectively. Then
$$\mathcal{G}_{\text{touch}} \subsetneq \mathcal{G}_{\text{low-int}} = \mathcal{G}_{\text{split}} = \mathcal{G}_{\text{c-sep}} \subsetneq
\mathcal{G}_{\text{int}}.$$

Before proving Lemma 1, we give a simple characterization of diagonal touching graphs that we call {\em crossing condition}, which was independently found by Hixon~\cite{Hixon13} and Soto and Thraves~\cite{MSoto13}.
\begin{Prop}\label{characterization_DG}
The diagonal touching graphs are the graphs such that there exists an order $<$ on the vertices, such that: for all $a,b,c,d \in V$ such that  $a<b<c<d$, if both $(a,c)\in E$ and $(b,d)\in E$ then $(b,c)\in E$.
\end{Prop}

\begin{proof}
Given a set of diagonal touching rectangles, we consider the ordering of the rectangles along the diagonal. Assume that the condition does not hold for four vertices $a<b<c<d$. As $(b,c)\notin E$, by symmetry we can assume that $u_x^b < \ell_x^c $. But $(a,c)\in E$ leads to $\ell_x^c \leq u_x^a  $, so $u_x^b < u_x^a$, contradicting $a<b$.

Consider now a graph with the property. We describe how to construct the rectangles in a way that their upper-right corner touches the diagonal. First put the top-right corners on the diagonal, in the ordering. Each rectangle will be just large enough to touch its furthest neighbors, i.e., for a rectangle $i$, if $k$ is the smallest neighbors (in the ordering) and $l$ the biggest, we choose $\ell^i_x = u^k_x$ and $\ell^i_y = u^l_y$. All the intersections in the graphs occur in this rectangle representation. Assume that there is an intersection between two rectangles $i<j$ and $(i,j)\notin E$. If the rectangle $i$ goes down enough to touch $j$ it means that there exists $k$, with $j<k$ and $(i,k)\in E$. By symmetry, there also exists $h$, with $h<i$ and $(h,j)\in E$. Then the crossing condition does not hold.
\end{proof}

\begin{Claim}
$\mathcal{G}_{\text{touch}} \subsetneq \mathcal{G}_{\text{low-int}}$
\end{Claim}

\begin{proof}
The inclusion of the two classes is a consequence of the geometric definitions of the classes. To prove that $\mathcal{G}_{\text{touch}}$ and
$\mathcal{G}_{\text{low-int}}$ are different we exhibit a specific graph in $\mathcal{G}_{\text{low-int}}\setminus \mathcal{G}_{\text{touch}}$. Before
doing that note that $\mathcal{G}_{\text{low-int}}$ is closed under adding a universal vertex, i.e., if $G=(V,E)$ is in
$\mathcal{G}_{\text{low-int}}$, then $\hat{G}=(V\cup \{u\},E \cup \{(v,u)|v\in V\})$ is also in the class (however this is not true for graphs in
$\mathcal{G}_{\text{touch}}$). Indeed, in $\mathcal{G}_{\text{low-int}}$ one can always create a rectangle $R$ with $a^R <a^1$, $b^R > b^n$,
so that $R$ is dominated by all other rectangles.
This new rectangle will be a universal vertex of the underlying graph.

\begin{figure}
\centering
\begin{tabular}{cc}
\includegraphics[scale=0.3]{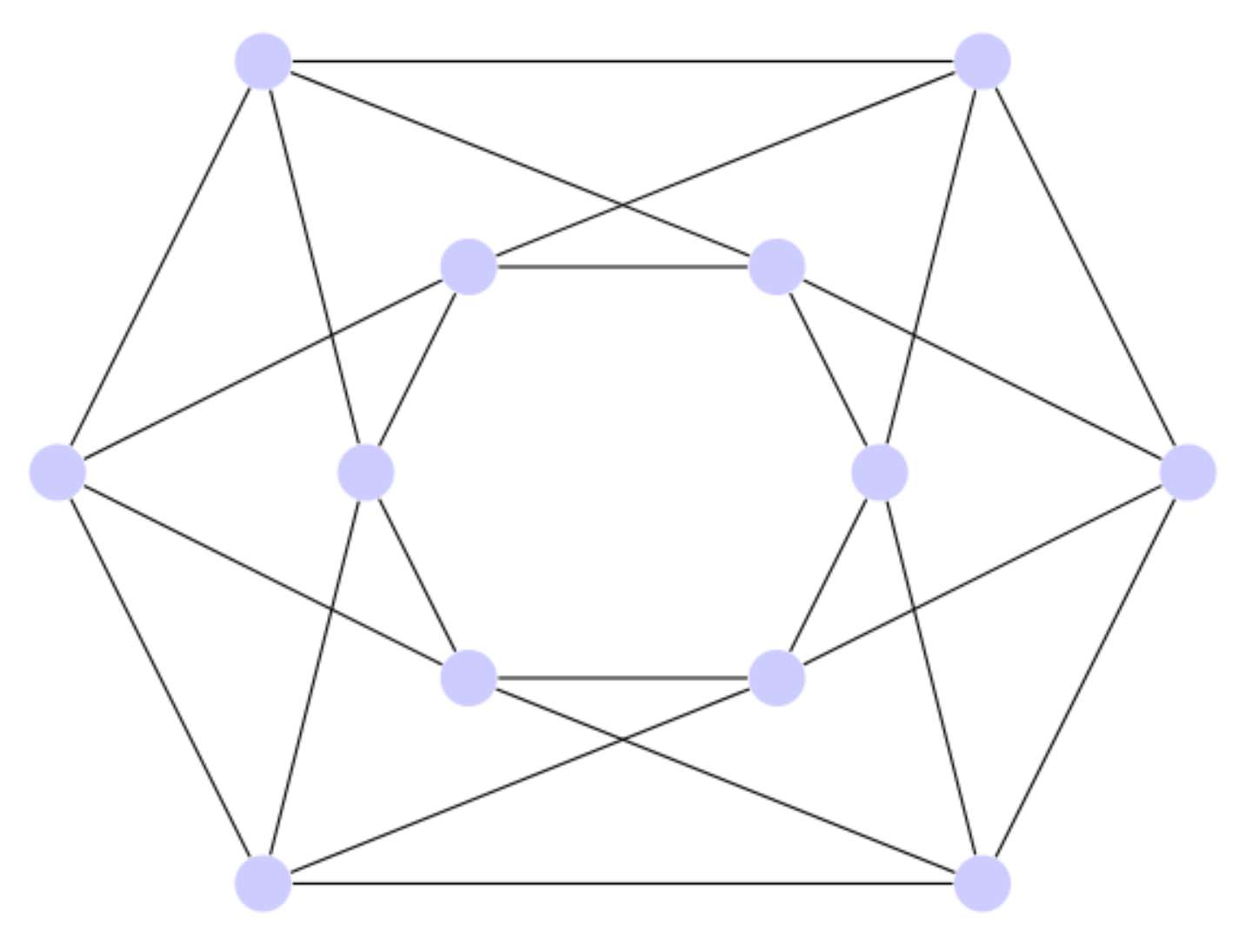}
&
\includegraphics[scale=0.6]{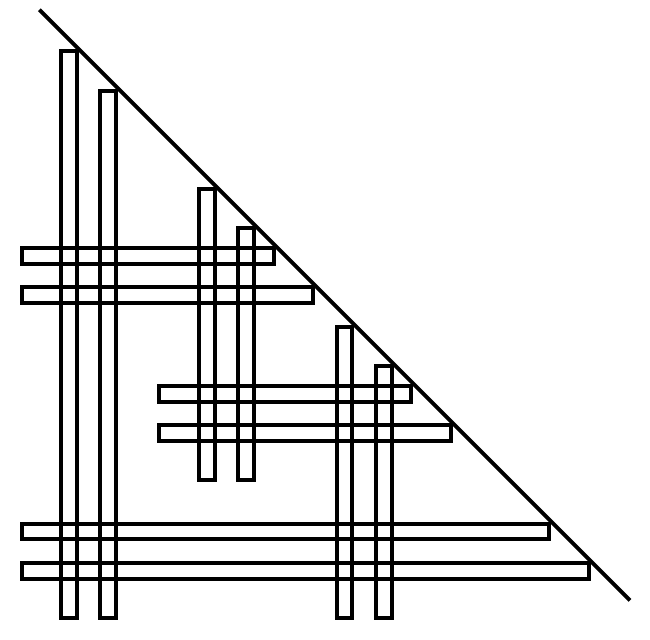}
\end{tabular}
\caption{\label{fig:double_six}$G_{2C_6}$, the doubled 6-cycle and its diagonal touching representation.}
\end{figure}

Thus consider the graph $G_{2C_6}$ of Figure \ref{fig:double_six}. This graph is clearly in $\mathcal{G}_{\text{touch}}$ (see the rectangle representation on the right of the figure). We show that if a universal vertex is added to $G_{2C_6}$, the graph is not in $\mathcal{G}_{\text{touch}}$ anymore, while by the previous observation it certainly belongs to $\mathcal{G}_{\text{low-int}}$.

First remark that in diagonal touching position, the rectangle that corresponds to a universal vertex ($i$ in the ordering), define a partition of the vertices $\{r|r<i\} \cup \{r|r>i\}$ inducing two interval graphs. Indeed as all the rectangles of $\{r|r<i\}$ touch the upper side of $i$ they form an interval graph; similarly for $\{r|r<i\}$ with the right side. The property is illustrated by Figure \ref{fig:universal_vertex}.

\begin{figure}
\centering
\includegraphics[scale=0.9]{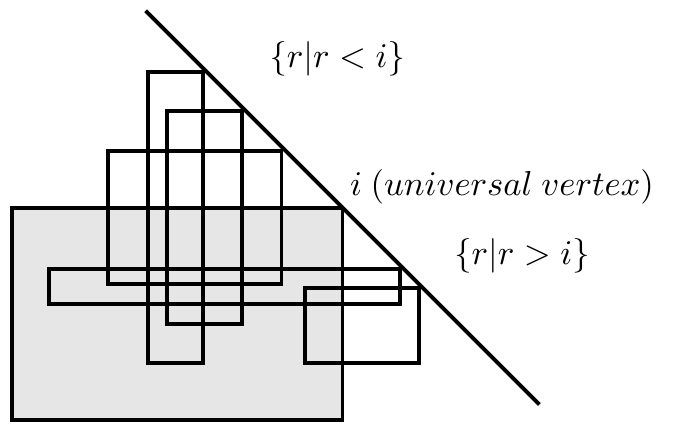}
\caption{\label{fig:universal_vertex} Universal vertex and partition of the other rectangles.}
\end{figure}

For convenience we redraw $G_{2C_6}$ as in Fig. \ref{fig:tube} (keeping in mind that there are edges linking the two ends). Note that two vertices of the same column are topologically equivalent as they share exactly the same neighbors. In the diagonal representation, they corresponds to the {\em twin rectangles}.
\begin{figure}
\centering
\includegraphics[scale=0.3]{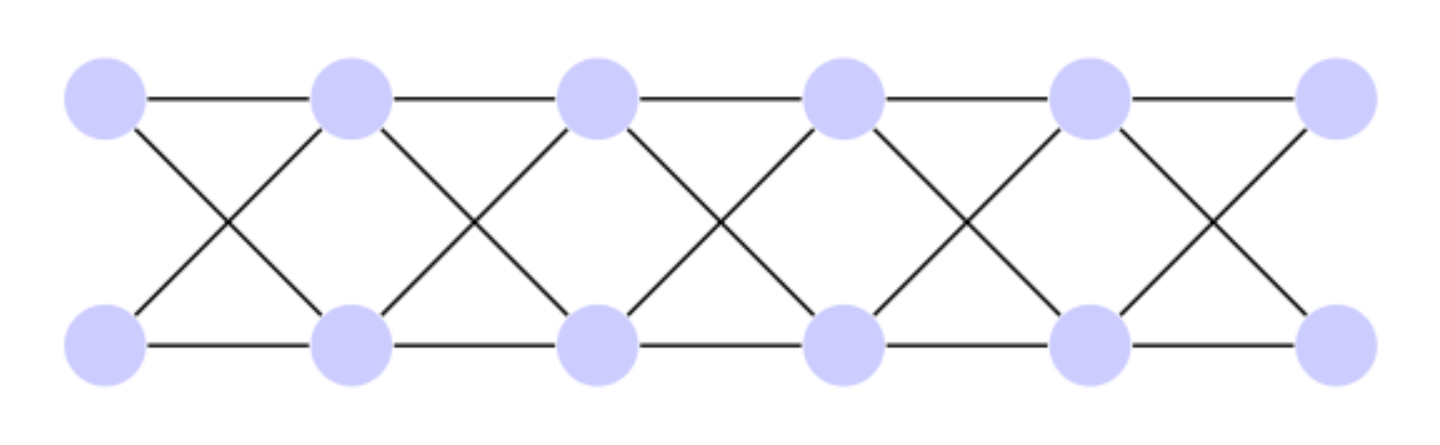}
\caption{\label{fig:tube} Tube representation.}
\end{figure}
Assume that a universal vertex can be added to $G_{2C_6}$ while staying in $\mathcal{G}_{\text{touch}}$. Then the rectangles can be partitioned into
two interval graphs. We say that a vertex is black (resp. white), if it is in the first (resp. second) interval graph. An alternating chain is a path
in the graph $G_{2C_6}$ such that two neighbors in the path have different colors and no two twins can be in the path. We consider the length $L$ of
a maximum alternating chain in the graph, and for each $L$ ($1 \leq L \leq 6 $) we show a contradiction. For this recall that an interval graph
cannot have a 4-cycle.

\begin{itemize}
\item[$L = 1$.] In this case there is only one color, say white, thus there is a white 4-cycle.
\item[$L = 2$.] In this case we proceed in three steps, see Figure \ref{fig:L=2}.
Take a maximum alternating chain (first step), as it is maximum the colors of the column on the left and on the right are determined (second step in the figure), then there is only one possibility to avoid the 4-cycles (third step). Implying in any case that there is an alternating chain of length 4, which is a contradiction.

\begin{figure}
\centering
\begin{tabular}{ccc}
\includegraphics[scale=0.8]{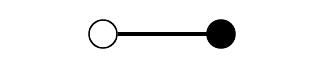}
&
\includegraphics[scale=0.8]{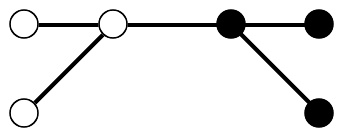}
&
\includegraphics[scale=0.8]{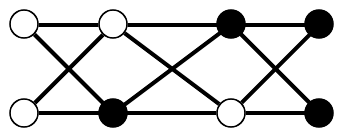}
\end{tabular}
\caption{\label{fig:L=2}The three steps for L=2}
\end{figure}

\item[$L=3$.] The argument is analogous to that of $L=2$.

\item[$L = 4$.] In the case the coloring is uniquely determined (up to the obvious color switching), and it is illustrated in  Figure \ref{fig:L=4_3_steps}.
\begin{figure}
\begin{center}
\begin{tabular}{ccc}
\includegraphics[scale=0.8]{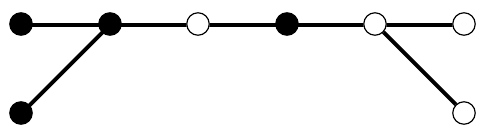}
&
\includegraphics[scale=0.8]{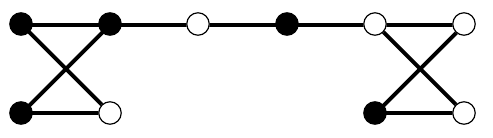}
&
\includegraphics[scale=0.8]{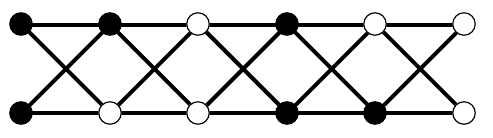}
\end{tabular}
\end{center}
\caption{\label{fig:L=4_3_steps}The three steps for L=4.}
\vspace{5ex}
\centering
\includegraphics[scale=0.8]{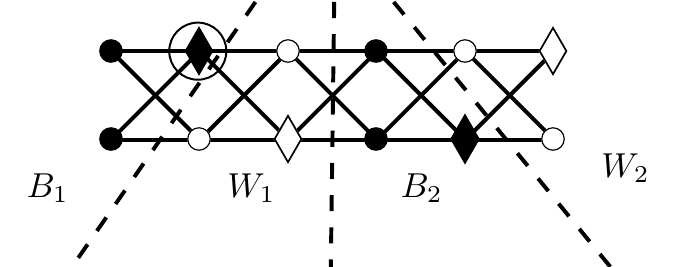}
\caption{\label{fig:L=4_groups}The case L=4, with the groups.}
\end{figure}

Observe that we have two independent groups of black vertices and two independent groups of white vertices (of course the colors are interchangeable), and in the ordering induced by the diagonal all vertices of one color are followed by all vertices of the other. Then in the ordering there is first the three vertices of the first black group $B_1$, then the other black group $B_2$, then a white group $W_1$ and then the other white group $W_2$. Note that by symetry the situation is fully equivalent to exchanging the role of $B_1$ and $B_2$, and/or that of $W_1$ and $W_2$ in the ordering. If follows that there exists $i\in B_1$,  $j \in B_2$, $k \in W_1,$ $l \in W_2$, such that $(i,k)\in E,$ $(j,l)\in E$ and $(k,l)\notin E$. (See for example the diamonds on Figure \ref{fig:L=4_groups}.) Then the graph is not diagonal because the crossing condition is violated, and therefore we obtain a contradiction.
\item[$L= 5$.] If there is an alternating chain of length 5 then there are two vertices of the same color at the end, and we have a 4-cycle like in Figure \ref{fig:L=5}, obtaining a contradiction.
\begin{figure}[!h]
\centering
\includegraphics[scale=0.8]{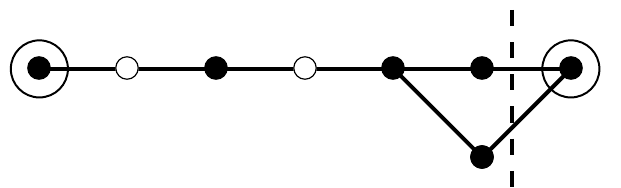}
\caption{\label{fig:L=5} The case L=5.}
\end{figure}
\item[$L = 6$.] We consider the induced subgraph with just the maximum chain, which has to be a cycle. This graph must be a diagonal-touching graph with an ordering having first three black vertices and then three whites ones. By inspection one can easily check that this is not possible.
\end{itemize}
Then $G_{2C_6}$ does not accept a universal vertex.
\end{proof}

\begin{Claim}
	$\mathcal{G}_{\text{c-sep}} =  \mathcal{G}_{\text{split}} = \mathcal{G}_{\text{low-int}} $.
\end{Claim}

\begin{proof}
First observe that $\mathcal{G}_{\text{c-sep}} \subseteq \mathcal{G}_{\text{split}}$, because to separate the top-right corners, the diagonal must intersect the upper sides of the rectangles. Also $\mathcal{G}_{\text{split}} \subseteq \mathcal{G}_{\text{low-int}}$, because if two rectangles $i<j$ in splitting position intersect above the diagonal, the point $a_j$ is also in the intersection, and this point is on the diagonal. Finally, if all the intersections of the rectangles are present below the diagonal, one can replace the top-right corner of a rectangle $i$ by $(b^i_x,a^i_y)$. This transformation does not change the intersection graph, as the parts of the rectangles below the diagonal do not change, and it does not create new intersections. The new rectangles are in corner-separated position. Then $\mathcal{G}_{\text{low-int}} \subseteq \mathcal{G}_{\text{c-sep}}$.
\end{proof}

\begin{Claim}
$\mathcal{G}_{\text{c-sep}} \subsetneq
\mathcal{G}_{\text{int}}$
\end{Claim}
\begin{proof}
The inclusion of the two classes is again a consequence of the geometric definitions of the classes. We now prove that the cube, $G_{\text{cube}}$, depicted in Figure \ref{fig:cube}, is in $\mathcal{G}_{\text{int}}$, but not in $\mathcal{G}_{\text{low-int}}$. The first assertion follows directly from the figure on the right. In what follows we prove the latter.
\begin{figure}[!h]
\centering
\begin{tabular}{cc}
\includegraphics[scale=0.25]{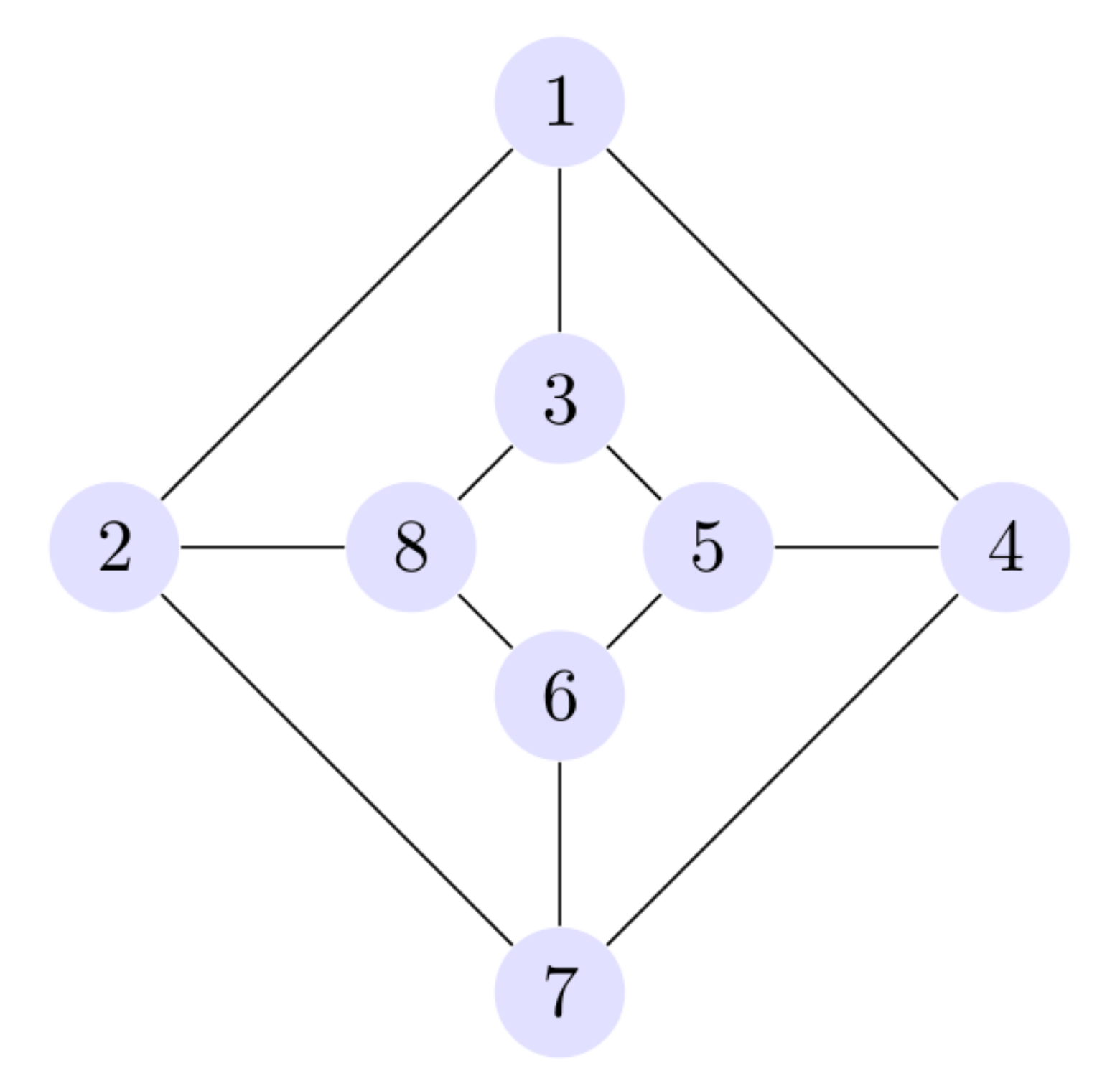}
&
\includegraphics[scale=0.8]{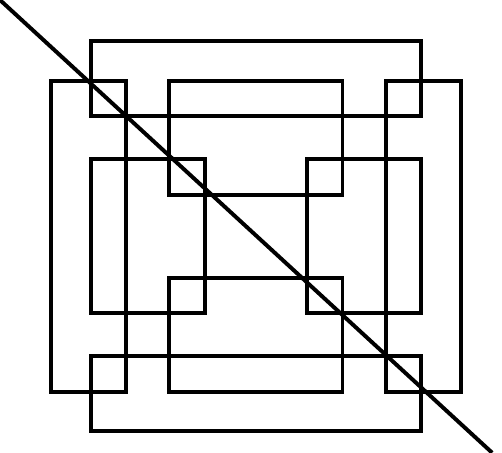}
\end{tabular}
\caption{ \label{fig:cube} The cube graph and diagonal intersecting representation.}
\end{figure}
Assume that there is a representation of $G_{\text{cube}}$ in lower-intersecting position and consider the ordering $<$ of the points $a_i$ along the diagonal. We say that a rectangle $v$ is nested in a rectangle $u$ if $a^u < a^v < b^v < b^u$. We now remark four basic properties about the graphs in $\mathcal{G}_{\text{low-int}}$ and their rectangle representations.

\begin{itemize}
\item[{\bf R1}] If $i<j<k$ and $i$ is nested in $k$, then $i$ and $j$ intersect.
\item[{\bf R2}] If there exists $i<j<k<l$ with $(i,k)\in E$, $(j,l)\in E$ and $(j,k)\notin E$ (i.e., the crossing pattern) then $j$ is  nested in $i$ or $l$ is nested in $k$. Indeed, if the crossing pattern is present there exists a path in the plane, below the diagonal, from $a_i$ to $a_k$ inside  rectangles $i$ and $k$, and a path from $a_j$ to $a_l$ inside $j$ and $l$, and these paths intersect. As $j$ and $k$ do not intersect, the intersection has to be in $i\cap j$ or $i \cap l$ or $k \cap l$. In the first and second cases, $j$ is nested in $i$, in the third case, $l$ is nested in $k$.
\item[{\bf R3}] It is not possible to have $i<j<k<l$ with: $(i,k) \in E$, $(j,l)\in E$, $(i,j)\notin E$, $(j,k)\notin E$ and $(k,l) \notin E$ (corollary of the previous remark).
\item[{\bf R4}] It is not possible to have $i<j<k<l$ with: $(i,l) \in E$, $(j,k)\in E$, $(i,k)\notin E$ and $j$ nested in $i$. Indeed, if $j$ is nested in $i$, then as $(j,k)\in E$ and $(i,k)\notin E$, we have $\ell^i_y > a^k_y$, but $(i,l)\in E$ so $\ell^i_y\leq a^l_y$, which is not possible as $a^k_y > a^l_y$.
\end{itemize}
Consider now the cube with the vertices named like in Figure \ref{fig:cube}. By symmetry, we may assume that the first vertex is 1 and that its neighbors, 2,3 and 4, appear in that ordering ($2<3<4$). Then we consider the different cases for vertex 7.

\begin{itemize}
\item If 7 is before 3 in the ordering (i.e., just after 1 or between 2 and 3), then (1,7,3,4) contradicts {\bf R3}.
\item If 7 is between 3 and 4. Using  {\bf R3} on (1,2,3,7), 2 must be nested in 1. Then (1,2,7,4) contradicts {\bf R4}.
\item If 7 is at the end, there is no contradiction if 2 is nested in 1. Thus we consider the possible positions of vertex 8 in the ordering: If $1 < 8 < 2$ then (1,8,2) contradicts {\bf R1}; If $2 < 8 < 3$ (resp. $3 < 8 < 4$)  then (1,2,8,3) (resp. (1,2,8,4)) contradicts  {\bf R4}; If $4 < 8 < 7$ (resp.  $7 < 8$) then (3,4,8,7) (resp. (2,3,7,8)) contradicts {\bf R3}. This covers all possible positions of vertex 8 in the ordering, obtaining a contradiction in each one.
\end{itemize}
Then the cube is not in $\mathcal{G}_{\text{c-sep}}$ and the classes are different.
\end{proof}

\section{Discussion}

To conclude the paper we mention open problems that are worth further investigation. First, note that the computational complexity of \MHS is open
for all classes of rectangle families considered in this paper. The complexity of
recognizing the intersection graphs of different rectangles families is also open. It is known that the most general version of this problem, that is
recognizing if a graph is the intersection graph of a family of rectangles, is NP-complete~\cite{Yannakakis82}. However, little is known for
restricted classes. Finally, it would be interesting to determine the duality gap for the classes of rectangle families studied here.

\subsection*{Acknowledgements} We thank V{\'i}t Jel{\'i}nek for allowing us to include the lower bound example in Figure~\ref{fig:gap2}, and Flavio Gu\'i\~nez and Mauricio
Soto for stimulating discussions.
This work was partially supported by
N\'ucleo Milenio Informaci\'on y Coordinaci\'on en Redes ICM/FIC P10-024F and done while the second author was visiting Universidad de Chile.

\end{document}